\documentclass[%
 reprint,
nofootinbib,
 amsmath,amssymb,amsthm,
 aps,
onecolumn
]{revtex4-2}

\usepackage{graphicx}
\usepackage{dcolumn}
\usepackage{caption}
\usepackage{bm}
\usepackage{float}
\usepackage{tcolorbox}
\usepackage{tikz}
\usepackage{wrapfig}
\usepackage{mathtools}
\input{insbox}
\usepackage[breaklinks]{hyperref}

\newcommand{\SO}{{\mathsf{SO}}}

\newcommand{\gn}{{\mathsf{n}}}
\newcommand{\weyl}{{\mathsf w}}
\newcommand{\sigweyl}{\weyl}
\newcommand{\elweyl}{{\mathsf s}}
\newcommand{\helweyl}{\hat\elweyl}
\newcommand{\hweyl}{\hat\weyl}
\newcommand{\Weyl}{{\mathsf W}}
\renewcommand{\sl}{{\mathfrak{sl}}}
\newcommand{\gl}{{\mathfrak{gl}}}
\newcommand{\so}{{\mathfrak{so}}}

\newcommand{\qPS}{\psi}
\newcommand{\QPS}{\Psi}
\newcommand{\rank}{r}
\newcommand{\spa}{u}
\newcommand{\multiplicity}{d}

\newcommand{\Yangian}{\mathcal{Y}}

\newcommand{\qrep}{W}
\newcommand{\grep}{\nu}

\newcommand{\SourcePS}{S}

\newcommand{\qV}{V}
\newcommand{\QV}{V}
\newcommand{\qv}{v}

\newcommand{\Qq}{{\bf P}}

\newcommand{\obasis}{\varepsilon}
\newcommand{\algg}{{\mathfrak{g}}}

\newcommand{\fullset}{{\bar{\emptyset}}}

\newcommand{\ii}{{\mathsf{i}}}

\newcommand{\qDirac}{\psi}
\newcommand{\QDirac}{\Psi}

\newcommand{\Projp}{\Gamma^+}
\newcommand{\Projm}{\Gamma^-}
\newcommand{\Projpm}{\Gamma^{\pm}}
\newcommand{\dBasis}{\theta}

\newcommand{\degQ}{M}

\newcommand{\Plucker}{Pl\"{u}cker }

\newcommand{\sh}{D}

\newcommand{\hQ}{{Q}}
\newcommand{\hQDirac}{{\QDirac}}
\newcommand{\hqV}{{\qV}}
\newcommand{\hQV}{{\QV}}

\newcommand{\hQq}{{\Qq}}
\newcommand{\bP}{{\hQq}}
\newcommand{\Pso}{\Phi}
\newcommand{\PSes}{{\QPS_{\emptyset}}}
\newcommand{\Tg}{T^{\gl_r}}
\newcommand{\bTg}{\mbox{\ensuremath{\bar T}}^{\gl_r}}
\newcommand{\es}{\emptyset}
\newcommand{\mB}{\bar{\mathcal{B}}}
\newcommand{\bmB}{\mathcal{B}}
\newcommand{\ba}{{\bf a}}
\newcommand{\bb}{{\bf b}}

\newcommand{\glQ}{{\mathcal{Q}}}
\newcommand{\basvec}{{\mathfrak{e}}}
\DeclareMathOperator{\TS}{TS}

\newcommand{\fs}{{\bar\es}}

\newcommand{\be}{\begin{eqnarray}}
\newcommand{\ee}{\end{eqnarray}}
\newcommand{\dressing}{\sigma}

\newcommand{\CO}{{\mathcal{O}}}

\newcommand{\genericdel}[4]{%
  \ifcase#3\relax
  \ifx#1.\else#1\fi#4\ifx#2.\else#2\fi\or
  \bigl#1#4\bigr#2\or
  \Bigl#1#4\Bigr#2\or
  \biggl#1#4\biggr#2\or
  \Biggl#1#4\Biggr#2\else
  \left#1#4\right#2\fi
}

\newcommand{\eg}{{\it e.g. }}
\newcommand{\ie}{{\it i.e. }}
\newcommand{\cf}{{\it cf. }}

\newcommand{\rhs}{{r.h.s. }}
\newcommand{\lhs}{{l.h.s. }}
\newcommand{\wrt}{{w.r.t. }}

\usepackage{amsthm}
\newtheorem{theorem}{Theorem}

\newtheorem*{conjecture*}{Conjecture}
\newtheorem{corollary}{Corollary}[theorem]
\newtheorem{lemma}[theorem]{Lemma}
\newtheorem*{lemma*}{Lemma}

\theoremstyle{definition}

\theoremstyle{remark}

\begin{document}
\preprint{APS/123-QED}

\title{Bethe Algebra using Pure Spinors}

\author{Simon Ekhammar}
\email{simon.ekhammar@physics.uu.se}
\author{Dmytro Volin}%
 \email{dmytro.volin@physics.uu.se}
\affiliation{Department of Physics and Astronomy,
Uppsala University, Box 516, SE-751 20 Uppsala, Sweden}
\affiliation{Nordita, KTH Royal Institute of Technology and Stockholm University,
Hannes Alfvéns väg, SE-106 91 Stockholm, Sweden}

\begin{abstract}
We explore a $\gl_{\rank}$-covariant parameterisation of Bethe algebra appearing in $\so_{2\rank}$ integrable models, demonstrate its geometric origin from a fused flag, and use it to compute the spectrum of periodic rational spin chains, for various choices of the rank $r$ and Drinfeld polynomials.
\end{abstract}

\maketitle

\tableofcontents
\section{Introduction}
Conserved charges of integrable models form a commutative algebra known as the Bethe algebra. In this paper, we propose an efficient way to parameterise it for models with $\so_{2\rank}$ symmetry. Our motivation is two-fold. From a practical perspective, $\so_\gn$ systems emerged in AdS/CFT integrability offering an elementary toolkit to study holography, in terms of fishnets and fishchains \citep{Gurdogan:2015csr,*Gromov:2019bsj}. Understanding of the Bethe algebra is required for solving thermodynamic Bethe Ansatz (TBA) equations and performing separation of variables \citep{Balog:2005yz,*Basso:2019xay,*Derkachov:2019tzo,*Cavaglia:2021mft}. On the other hand, from a conceptual point of view, the Bethe algebra is surprisingly universal. For generic values of the spectral parameter, it features the same relations that appear in numerous studies: integrable models derived by quantum inverse scattering method (with rational, trigonometric, elliptic cases included); TBA (both relativistic and AdS/CFT-type systems); a variety of differential equations (leading to ODE/IM and similar correspondences); supersymmetric gauge theories (by virtue of Bethe/gauge correspondence); enumerative geometry (where the Bethe algebra is a quantum cohomology ring) \citep{[{See e.g. the following works and references therein: }][] Kuniba:2010ir,*Masoero:2015rcz,*Nekrasov:2013xda,*Maulik:2012wi}. Another name for our study is parameterisation of a finite-difference oper, \eg of the one in \cite{Frenkel:2020iqq}. Whereas universality holds for any Lie algebra $\algg$, $\so_{2\rank}$ is the simplest possible case after $\sl_{\rank+1}$. The $\sl_{\rank+1}$ case, although studied in great detail, is liable to simplifications shading the nature of relations we would like to explore. $\so_{2\rank}$ is hence an important stepping stone towards the general case we aim to address in future works.

In \cite{Ekhammar:2020enr} building on techniques of \cite{Sun:2012xw,Masoero:2015lga}, H.~Shu and the authors described the Bethe algebra for any simply-laced $\algg$ in terms of an extended Q-system. This approach offers full covariance: if an integrable model is based on the quantum deformation of the current or the loop algebra of $\algg$ then Baxter Q-functions generating the Bethe algebra are \Plucker coordinates of the fused flag corresponding to the Langlands dual $^L\hat\algg$. A physical counterpart is a recently proposed realisation of Q-operators by magnetically charged 't Hooft lines \cite{Costello:2021zcl}.

The covariant description has an unfortunate downside: too many Q-functions are involved. Indeed, the number of Q-functions of the extended Q-system grows exponentially with $\rank$, whereas the number of functionally independent Q-functions is equal to the rank of the algebra. 

In the $\sl_{\rank}$ case (of rank $\rank-1$), this issue is elegantly handled with `single-indexed' Q-functions $\glQ_a$, $a=1,\ldots,\rank$, which are projective coordinates $[\glQ_1:\ldots:\glQ_\rank]$ of $\mathbb{P}^{\rank-1}$. They naturally transform under $\sl_{\rank}$-action ($\mathsf{PGL}_r$ to be precise, but we mostly work on the Lie algebra level). The other Q-functions, which are \Plucker coordinates of Grassmannians $\mathsf{Gr}(\rank,k)$,  are conveniently computed by discrete Wronskian determinants $\glQ_{a_1\ldots a_k}=W(\glQ_{a_1},\ldots,\glQ_{a_{k}})$ \citep{Krichever:1996qd,Tsuboi:2009ud}. For a reader familiar with nested Bethe Ansatz, we recall that Q-functions in the case of spin chains have the structure $\glQ_{a_1\ldots a_k}=\sigma_k\, q_{a_1\ldots a_k}$, where $q_{a_1\ldots a_k}$ are polynomials. A rather standard notation choice is to identify $q_{1\ldots k}$ for $k=1,\ldots,r-1$ as the polynomials whose zeros satisfy nested Bethe equations. We can then choose $\glQ_{12\ldots k}$ to generate the Bethe algebra but this is probably the least covariant way to proceed, and this is not what we are aiming for. We want to work with single-indexed $\glQ_a$ as generators, and this is practically feasible because a replacement of nested Bethe equations exists: Wronskian Bethe equations $W(\glQ_1,\glQ_2,\ldots, \glQ_\rank)=1$ are used directly for fixing $\glQ_a$.

For  $\so_{2\rank}$ models, candidates for analogs of $\glQ_a$ appeared in the work \cite{Ferrando:2020vzk} by G.~Ferrando, R.~Frassek, and V.~Kazakov who considered Q-system on the Weyl orbit (we use the terminology of \cite{Ekhammar:2020enr}).  Overall, this work offered a variety of relations that transform naturally under the action of the Weyl group of $\so_{2\rank}$ and moreover are $\gl_\rank$-covariant, where $\gl_r\subset \so_{2\rank}$ corresponds to the standard embedding $\mathsf{U}(\rank)\subset \SO(2\rank)$ (we do not need to specify real forms and hence select $\gl_r$ for notation). In particular, equation (7.2) there  was suggested as a substitute for Wronskian Bethe equations, it featured $\rank+1$ Q-functions, their equivalents in our work shall be $\Psi_a$ and $\Psi_{\emptyset}$, where $\Psi_a$ form a vector $\gl_r$-multiplet and $\Psi_{\emptyset}$ is a singlet. However, unlike in the $\sl_{\rank}$ case, (7.2) is not sufficient by itself. As we shall see, extra relations involving an anti-symmetric tensor $\Psi_{ab}$ are required.

In this paper we explain how the $\gl_r$-covariant expressions of \cite{Ferrando:2020vzk} are related to the geometric $\so_{2r}$-covariant description in \cite{Ekhammar:2020enr}. We investigate this interplay and complete (7.2) to a system of equations on a relatively small number of Q-functions to effectively encode the whole Bethe algebra.

Our study starts from the observation made in \cite{Ekhammar:2020enr}: Because Q-functions are \Plucker coordinates, those of them who define maximal isotropic subspaces are components of pure spinors. Then one can benefit from Cartan parameterisation of pure spinors which uses only $1+\rank+{r\choose 2}$ functions, we shall denote them respectively as
\be
\PSes\equiv 1/\Pso\,,\quad \QPS_a \equiv \Qq_a/\Pso\,,\quad \QPS_{ab}\equiv\mu_{ab}/\Pso\,.
\ee
Because of the fusion relations, only $\rank$ of them are independent. If one chooses $\Qq_a$ as the independent ones, they would be naturally the `single-indexed' Q-functions of a $\gl_r$ Q-system (and not of $\sl_r$ because they alone do not form a set of projective coordinates). Owing to other fused flag relations, all the members of the extended Q-system are just rational combinations of $\PSes,\QPS_a,\QPS_{ab}$ and so this triple contains all, or almost all, non-trivial information about the Bethe algebra. A suitable approach to work with the triple depends on the question to study, we give two examples in Sections~\ref{sec:Tfunc}~and~\ref{sec:RationalSpinChains}.

The paper is organised as follows: in Section~\ref{sec:FusedFlags} we summarise the relevant findings of \cite{Ekhammar:2020enr} about the extended $\so_{2r}$ Q-system, with a slight update of notations. Section~\ref{sec:PSparameterisation} describes $\gl_\rank$ parameterisation of the Q-system and builds up to relations \eqref{eq:SpinorQSystem} between $\PSes,\QPS_a,\QPS_{ab}$ which play the role of Wronskian Bethe equations and is a way to concisely describe Bethe Algebra. From these relations we then derive all the features of the extended Q-system and use this techniques to prove that, starting from generic enough $r$ Q-functions as an input, one can always construct a consistent extended Q-system in the unique up to the symmetry way.   Section~\ref{sec:Tfunc} studies a $\gl_r$-invariant approach to compute transfer matrices in terms of $\Phi, \Qq_a$ only but eventually resolves that adding $\mu_{ab}$, while optional, is a useful simplification. Section~\ref{sec:RationalSpinChains} is a proof of concept that the developed formalism works in practice: we explicitly and efficiently compute the spectrum of rational spin chains, the results are available in the ancillary {\it Mathematica} notebook. These results offer rich experimental evidence that equations \eqref{eq:SpinorQSystem}, plus kinematic constraints in the non-basic cases, provide a rigorous description of the Bethe Algebra in the same sense as $\sl_{\rank}$ Wronskian Bethe equations do, the precise mathematical conjecture is in Section~\ref{sec:CF}. The results of the work are summarised in section Conclusions. The first appendix explores action of Weyl group on Q-functions and provides further comparison with the results of \cite{Ferrando:2020vzk}. Finally, the second appendix contains several technical proofs.
\section{\label{sec:FusedFlags}Extended \texorpdfstring{$\so_{2\rank}$}{so(2r)} Q-system}
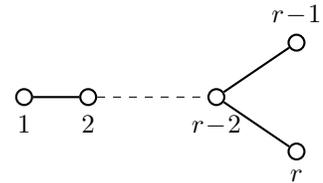
\begin{wrapfigure}{r}{0.3\textwidth}
{
\begin{center}
            \begin{tikzpicture}[scale=.426] \tikzset{node/.style={draw,circle,thick,inner sep=0pt,minimum size=6pt}}
                \node[node, label=below:$1$] (n1) at (0, 0) {};
                \node[node, label=below:$2$] (n2) at (2, 0) {};
                \node[node, label=below:$r\!-\!2$] (n3) at (6, 0) {};
                \node[node, label=above:$r\!-\!1$] (n4) at (8.5, 1.7) {};
                \node[node, label=below:$r$] (n5) at (8.5, -1.7) {};
                \draw[thick] (n1) -- (n2);
                \draw[dashed](n2)--(n3);
                \draw[thick] (n3) -- (n4);
                \draw[thick] (n3) -- (n5);
            \end{tikzpicture}
       \caption{Labeling conventions on Dynkin diagram}
    \label{fig:DDynkin}
\end{center}
}
\end{wrapfigure} 
The $\so_{2\rank}$ extended Q-system is a collection of functions $Q_{(a)}^\mathfrak{i}(u)$ of a single variable $u$ called the spectral parameter. These functions must satisfy a set of equations listed below in this section. Label $a$ denotes the choice of a fundamental representation $L(\omega_a)$, $a=1,2,\dots,\rank$, see our labelling convention on the Dynkin diagram, and $\mathfrak{i}$ runs over all components of the representation. With an appropriate choice \cite{Ekhammar:2020enr} of weight basis vectors $\basvec_{\mathfrak{i}}^{(a)}$ in $L(\omega_a)$, one forms a Q-vector $Q_{(a)}=\sum_{\mathfrak{i}}Q_{(a)}^{\mathfrak{i}}\basvec_{\mathfrak{i}}^{(a)}$.

Until Section~\ref{sec:RationalSpinChains}, our discussion will be of universal nature. For this sake, functions of $u$ need only to be defined and single-valued in a certain domain, their analyticity and even continuity may be in principle waived. What matters is existence of a free action of the `shift' group $\mathbb{Z}$ on the domain of definition. For a function $F$, the notation $F^{[k]}$, $k\in\mathbb{Z}$, shall denote the pullback through this action. We also use the short-hand $F^{\pm}\equiv F^{[\pm 1]}$.

In Section~\ref{sec:RationalSpinChains}, we study an explicit physical example of a rational spin chain when the analytic properties of Q-functions will play a role. For this example, the shift group lives up to its name: $F^{[k]}(u)=F(u+\frac{\ii}{2}k)$. Other typical case is a trigonometric system with the multiplicative shift $F^{[k]}(z)=F(q^{k/2}z)$ (we will not consider it). The additive implementation of the shift  explains the mnemonics of our abbreviations: The shift operator $\sh=e^{\frac \ii2\partial_u}$ acts as $D^k\,F=F^{[k]}$, the fused power is defined by $f^{[n]_\sh} =\prod\limits^{n}_{k=1} f^{[n+1-2k]}$, and the finite-difference Wronskian determinant is $\label{eq:WronskianDef}
    W(F_{1},F_{2},\dots,F_{k}) = |F^{[k+1-2a]}_{b}|_{\substack{a=1\dots k\\b=1\dots k}}$, where $|M_{ab}|$ is the determinant of a matrix, the range of indices $a,b$ will be omitted when there is no risk of confusion.

The function $Q_{(a)}(u)$ valued in $L(\omega_a)$ can be viewed locally as a section in an associated $\mathsf{SO}({2r})$-bundle (or, more accurately, $\mathsf{Pin}(2\rank)$-bundle). We can then introduce a connection and gauge the definition of the shift $Q_{(a)}^{[k]}$, it would be a finite-difference analog of passing from ordinary to covariant derivative, \cf \cite{Kazakov:2015efa} for a similar discussion in the $\gl_{\rank}$ case. This way of thinking links our results to those about opers, as is explained in \cite{Ekhammar:2020enr}, and has its advantages. However, one can always locally gauge away the connection as long as the function class of Q-functions is unconstrained, and  we assume this gauge choice throughout the paper.

\medskip
\noindent
Probably the most famous relations satisfied by Q-functions are the so-called QQ-relations \cite{Mukhin:2005aa},\cite{Masoero:2015lga}. If to include their transforms under action of Weyl group elements $\weyl$---precise definitions are in \cite{Ekhammar:2020enr} and Appendix~\ref{sec:Weyl}---they read
\be\label{eq:DynkinDiagramQQ}
W(\hQ_{(a)}^{\weyl(1)},\hQ_{(a)}^{\weyl(2)}) =\pm  \prod_{b\sim a} \hQ_{(b)}^{\weyl(1)}\,,
\ee
where $b\sim a$ indicates nodes on the Dynkin diagram adjacent to $a$, and the sign $\pm$ depends on $\weyl$. 

The other relations of the extended Q-system, at least in general position situation, ultimately follow from \eqref{eq:DynkinDiagramQQ} but this feature is not trivial to demonstrate, its full combinatorial proof is one of the results of this work. For the moment we shall list them as an additional requirement. Some of them first time appeared in \cite{Sun:2012xw} but their complete package is due to \cite{Ferrando:2020vzk},\cite{Ekhammar:2020enr}. These relations have an intriguing geometrical meaning. For instance, they imply that components of $Q_{(a)}$ are projective \Plucker coordinates $[Q_{(a)}^{1}:Q_{(a)}^{2}:\ldots]$ of the partial flag manifold $G/P_a$, where $P_a$ is the corresponding maximal parabolic subgroup. More generally, they tell us that all Q-functions together parameterise a fused flag \cite{Ekhammar:2020enr}. The very fact that all these relations can be simultaneously satisfied is remarkable.

\medskip
\noindent
To list down the relations, it is useful to introduce a more intuitive labelling for functions $Q_{(a)}^\mathfrak{i}$. Note that the vector representation of $\so_{2\rank}$ is associated to the first node of the Dynkin diagram and the spinor representations with the last two nodes. Hence the following notation becomes natural:
\begin{align}\label{eq:BasicRepresentations}
    &Q_{(1)}^{i} = \qV^{i}\,, 
    &
    &Q_{(\rank-1)} = \Projm\QDirac\,,   
    &
    &Q_{(\rank)} = \Projp\QDirac\,. 
\end{align}
Here $\QDirac$ is a $2^{\rank}$-dimensional Dirac spinor and $\Projp,\Projm$ are projection-matrices onto $2^{r-1}$-dimensional Weyl spinors. An $\so(2r)$ vector has $2\rank$ components, to label them it is convenient to use $i,j,\ldots\in \{1,2,\dots,\rank,-\rank,\dots,-2,-1\}$, while $a,b,\ldots\in \{1,2,\dots, \rank\}$ shall be always positive. To further precise conventions, we choose the anti-diagonal metric $g_{ij}$ and $\Gamma$-matrices forming the Clifford algebra:
\be\label{eq:TensorsConventions}
    g_{ij} = \delta_{i+j,0}\,,
    \quad
    \{\Gamma_{i},\Gamma_{j}\} = g_{ij}\,.
\ee
Also, we pick the charge-conjugation matrix $C$ as
\be\label{eq:ChargeC}
    C = (-1)^{\frac{r(r-1)}{2}}(\Gamma_{1}+\Gamma_{-1}) \dots  (\Gamma_{\rank}+\Gamma_{-\rank})\,.
\ee
Using the metric we have the identification
\begin{equation}
\label{eq:ama}
    V^{-a} = V_{a}\,.
\end{equation}

\begin{samepage}
Apart from Q-functions that are components for the vector and the spinor representations, there are $\rank-3$ other collections of Q-functions corresponding to anti-symmetric tensor representations. These Q-functions can be equated to Wronskians of $\qV^{i}$: Introduce a multi-index $I=i_1i_2\dots i_a$, and let $|I|=a$ be its cardinality, then, for $Q_{(a\leq r-2)}^{\mathfrak{i}\equiv I}\equiv V^{I}$,
\begin{equation}\label{eq:TensorQ}
    V^{I} = W(V^{i_1},\dots V^{i_a})\,.
\end{equation}
For $|I|=r-1,r$, \ie for tensors outside of the Dynkin diagram, \eqref{eq:TensorQ} shall be the definition of $V^{I}$. 
\end{samepage}

The above identification \eqref{eq:TensorQ} is analogous to the one in the $\gl_{\rank}$ Q-system~\footnote{We recall its definition in a paragraph after \eqref{eq:GammaAsForms}}, however $\so_{2\rank}$ Q-system has novel features. For one thing, tensors $V^{I}$ define isotropic hyperplanes for $|I|\leq r-1$ which means that the following orthogonality relation holds
\begin{equation}
\label{eq:fusedortho}
    \left(V^{i}\right)^{[m]}V_{i}^{[-m]}=0\,,\quad m\in \{-\rank+2,\dots,\rank-2\} \,;
\end{equation}
In particular, the vectors are null $V_iV^{i}=0$. For another, it is possible to build $V^{I}$ from spinors:
\begin{subequations}
\label{eq:VectorSpinorRelations}
\begin{align}
\label{eq:VectorSpinorRelationsa}
    &\hQDirac^{[-\rank+1+|I|]}C\Gamma^{I}\Projpm\hQDirac^{[\rank-1-|I|]}=
   \hqV^{I}\,,\quad
    |I|<\rank \,, \\
\label{eq:VectorSpinorRelationsb}
    &\hQDirac^{-}C\Gamma^{I}\hQDirac^{+} = \hqV^{I}\,,
    \quad
    |I| = \rank\,.
\end{align}
\end{subequations}
Spinors also satisfy relations analogous to \eqref{eq:fusedortho}. Namely, when the shift in \eqref{eq:VectorSpinorRelationsa} is smaller than the prescribed one, the bilinear combination vanishes
\begin{equation}\label{eq:ProjectionRelations}
    \QDirac^{[-m]}C\Gamma^{I}\Projpm\QDirac^{[m]}=
    0\,, \quad |I| \leq \rank-2\,, \quad m \in   \{-\rank+2+|I|,\dots,\rank-2-|I|\} \,. 
\end{equation}
In particular $\QDirac$ is a pure spinor, that is, it satisfies
\begin{equation}
\label{eq:PureSpinor}
    \QDirac C \Gamma^I \Projpm \QDirac = 0\,,
    \quad
    |I| \leq \rank-2\,.
\end{equation}
Relations of type \eqref{eq:ProjectionRelations} appeared for $\so_6\!\subset\! \mathfrak{osp}_{6|4}$ in the study of the AdS$_4$/CFT$_3$ quantum spectral curve~\cite{Bombardelli:2017vhk}. 

Finally, there is also a non-vanishing inner product between vectors and  between spinors when the appropriate shift of the spectral parameter is included:
\begin{subequations}
\label{eq:Quant}
\begin{align}
\label{eq:QuantVector}
    &\left(\hqV^{I}\right)^{[\rank-1]}\hqV_{I}^{[-\rank+1]} = (-1)^{|I|(\rank+1)}\,,
    \\
\label{eq:QuantSpinor}
    &\hQDirac^{[-\rank+1]} C\Projpm \hQDirac^{[\rank-1]} = 1\,. 
\end{align}
\end{subequations}
The values of shifts $\pm(\rank-1)$ here are related to the fact that $2\rank-2$ is the Coxeter number of $\so_{2r}$. Summation over a multi-index $I$ is defined to carry a normalization factor of $\frac{1}{|I|!}$, alternatively the sum is only over ordered multi-indices. The same summation convention shall be used below for multi-indices $A,B,\ldots $ that are defined to feature only positive entries.

Finally, to avoid misinterpretations, we spell out how the above-encountered expressions of type $\Psi_1 C \Gamma^I \Psi_2$ are decoded for the sake of explicit computations: $\Psi_1$ is transposed to be a row-vector (boldly, without doing extra complex conjugation or other involution), and then the standard row-times-column multiplication is performed to get a number which depends on $I$ and  which is the component of the corresponding rank-$|I|$ antisymmetric tensor.

\section{Pure spinor parameterisation}\label{sec:PSparameterisation}
\subsection{Cartan decomposition \label{sec:Cartandecomp}}
Consider the $\gl_\rank$ subalgebra of $\so_{2\rank}$ corresponding to removing the $r$-th node of the Dynkin diagram (but keeping the whole Cartan subalgebra) and decompose spinor representations into the irreps \wrt $\gl_\rank$. This gives rise to what is known as $\mathsf{U}(\rank)$ or Cartan decomposition. In it, the spinor representation of the $(r-1)$-th node is a direct sum of the exterior forms of odd rank, and the spinor representation of the $r$-th node is a direct sum of the exterior forms of even rank. We shall denote these forms collectively as $\QPS_{(k)}$, $k=0,\dots,\rank$, the components of the $k$-form $\QPS_{(k)}$ are $\QPS_{a_1\dots,a_k}$ and we denote the basis $\dBasis^{A}$ for $A=a_1\ldots a_k$.  We think about $\dBasis^a$ as Grassmann variables and Dirac spinor is treated as a super-function:
\begin{equation}
   \QDirac = \sum^{\rank}_{k=0} \QPS_{(k)}\,,
   \quad{\rm where}\quad
   \QDirac_{(k)}=\QPS_{A}\dBasis^{A}\,,\quad |A|=k\,.
\end{equation}
   Furthermore we will raise and lower indices using Levi-Civita symbol so that
$
    \QPS_{A} = \QPS^{B}\epsilon_{BA}\,. 
$

The off-diagonal metric in \eqref{eq:TensorsConventions} implies that the $\Gamma$-matrices satisfy the fermionic oscillator algebra. It can be realised on the $\dBasis$-basis as 
\begin{subequations}
\label{eq:GammaAsOperators}
\begin{align}
\label{eq:14a}
    &\Gamma^{a}\,\dBasis^{B} = \theta^a\dBasis^{B}\,, 
    \quad
    \Gamma_{a}\,\dBasis^{B} = \partial_a\dBasis^{B}\,,
    \\
    &\Projpm\,\dBasis^{B} = \frac{1\pm(-1)^{|B|}}{2}\dBasis^{B}\,,
    \\
    &C\,\theta^{A} = (-1)^{\frac{(r-|A|)(r-|A|-1)}{2}}\,*\dBasis^A\,,    
\end{align}
\end{subequations}
where $\bar{A}$ is the complement of $A$ and $*\dBasis^A\equiv \epsilon^{A\bar{A}}\dBasis^{\bar{A}}$ is Hodge conjugation. 

As $\mathfrak{so}_{2\rank}$ is realised by commutator of $\Gamma$-matrices, action  \eqref{eq:14a} suggests the decomposition $\so_{2\rank}=\Lambda^2(\mathbb{C}^\rank)\oplus\mathfrak{u}(\rank)\oplus\Lambda^2(\mathbb{C}^{*\rank})$, where the three terms are spanned, respectively, by $\theta^a\theta^b$, $\theta^a\partial_b-\frac 12\delta^a{}_b$, $\partial_a\partial_b$. Additionally, one should think about $\theta^a\partial_b$ with $a>b$ and $\partial_a\partial_b$ as raising operators, $\theta^a\partial_a-\frac 12$ for $a=1,\ldots,\rank$ as Cartan subalgebra, and the rest as lowering operators. It follows that the highest-weight components of the Q-vectors are
\begin{equation}
\label{eq:419}
    \PSes = Q_{(r)}^{1}\,,
    \quad
    \QPS_{\rank} = Q_{(r-1)}^{1}\,.
\end{equation}

In the introduced parameterisation,  $\Psi_1 C \Gamma^I \Psi_2$ (where $\Psi_1,\Psi_2$ are column-vectors) gets the meaning of the top component of the form $\Psi_1\wedge *\, C\Gamma^I \Psi_2$ (where $\Psi_1,\Psi_2$ are super-functions).

We shall now use the introduced exterior algebra notation to solve explicitly the pure spinor conditions. To this end, one computes
\begin{subequations}
\label{eq:448}
\begin{align}
\label{eq:448a}
0=\epsilon_{AB}\left(\QDirac\,C\Gamma^{A}\Gamma^{\pm}\QDirac\right)\dBasis^{B}  
&=
(-1)^{{r-|A|} \choose 2}\sum\limits_{\substack{k =0 \\k\ \rm{even/odd}}}^{\rank-|A|}(-1)^{\frac{k(k-1)}{2}}\QPS_{(\rank-|A|-k)} \wedge \QPS_{(k)}\,,& r-|A| &\geq 2\,,
\\
\label{eq:448b}
0=\epsilon_{AB}\left(\QDirac\,C\Gamma^{Ac}{}\Gamma_{c}\Gamma^{\pm}\QDirac\right)\dBasis^{B}  
&=
(-1)^{{r-|A|} \choose 2}\sum\limits_{\substack{k =0 \\k\ \rm{even/odd}}}^{\rank-|A|}(-1)^{\frac{k(k-1)}{2}}k\,\QPS_{(\rank-|A|-k)} \wedge \QPS_{(k)}\,,& r-|A|&\geq 4\,.
\end{align}    
\end{subequations}
Here in the second line we use a non-antisymmetrised expression $\Gamma^{Ac}\Gamma_{c}$ which is acceptable since its difference with $\Gamma^{Ac}{}_{c}$ is $\frac 12\Gamma^A$ for which the first line can be used. It is easy to solve \eqref{eq:448} explicitly  for all $\QPS_{(k)}$ with $k\geq 3$ in terms of $\Psi_{(k)}$ for $k=0,1,2$. We illustrate the first two steps. Taking even $k$ and $r-|A|=3$ in \eqref{eq:448a}, one gets $\Psi_{(3)}\Psi_{\es}-\Psi_{(1)}\wedge\Psi_{(2)}=0$, while  $r-|A|=4$ leads to $\Psi_{(2)}\wedge\Psi_{(2)}-2\,\Psi_{\es}\Psi_{(4)}=0$. For $r-|A|\geq 4$ it is sufficient to use only \eqref{eq:448b}. The full solution is
\begin{align}
\label{eq:PureSpinorCondition}
\boxed{
    \QPS_{(2n)} = \frac{1}{n!}\frac{\QPS_{(2)}\wedge \QPS_{(2)} \wedge \dots \wedge \QPS_{(2)}}{\PSes^{n-1}}\,,\quad \QPS_{(2n-1)} = \frac{1}{(n-1)!}\frac{\QPS_{(1)}\wedge \QPS_{(2)}\wedge \dots \wedge \QPS_{(2)}}{\PSes^{n-1}}
    }\,.
\end{align}
This is a well-known property of pure spinors that can be found already in the book of Cartan \cite{Cartan:104700}. Its meaning becomes transparent after we rewrite it as follows:
\be\label{eq:exp}
Q_{(r)}=\sum_{n=0}^{[r/2]}\QPS_{(2n)}=e^{\frac{\QPS_{(2)}}{\QPS_{\es}}}\,\QPS_{\es}\,,\quad Q_{(r-1)}=\sum_{n=0}^{[(r-1)/2]}\QPS_{(2n+1)}=e^{\frac{\QPS_{(2)}}{\QPS_{\es}}}\,\QPS_{(1)}\,.
\ee
Notice that $\frac{\QPS_{(2)}}{\QPS_{\es}}\in\Lambda^2(\mathbb{C}^\rank)$ and so its exponent can be viewed as a group element.  Hence, since $\Psi_{(1)}\in e^{\mathfrak{u}(\rank)}\theta^\rank$, \eqref{eq:exp} states that Weyl spinors $Q_{(r)},Q_{(r-1)}$ are on the group orbit of the corresponding highest-weight vectors. Being on such an orbit is the universal definition of a pure vector for any choice of a Lie group and its irrep. Furthermore, a group element acting on the two highest-weight vectors of $L(\omega_{r-1})$ and $L(\omega_{r})$  to produce $Q_{(r)}$, $Q_{(r-1)}$ can be chosen to be the same one, this is just saying differently that  two pure Weyl spinors are combined into one pure Dirac spinor~$\Psi$. 
 
 This interpretation allows us to also argue why all relations \eqref{eq:PureSpinor} are satisfied even though we only used their subset to derive \eqref{eq:PureSpinorCondition}. Indeed, the two highest-weight vectors obviously satisfy \eqref{eq:PureSpinor} and then, by covariance, any spinors on the group orbit of the highest-weight vectors will satisfy \eqref{eq:PureSpinor} as well. A direct combinatorial derivation of \eqref{eq:PureSpinor} from \eqref{eq:PureSpinorCondition} shall be given later as part of Section~\ref{sec:DerivationOfExtendedQsystem}.

\subsection{Tensors from spinors}
We concluded in the previous subsection that all components of spinor Q-functions can be expressed through $\Psi_{\es},\Psi_{a},\Psi_{ab}$.
We shall now derive that all tensor Q-functions $V^{I}$ can be expressed through $\Psi_{\es},\Psi_{a},\Psi_{ab}$ as well. This time, in contrast to pure spinor relations \eqref{eq:PureSpinor}, we will use equations that are non-local in spectral parameter. In other words, we need the fact that the extended Q-system is not just a flag but a fused flag.

A typical fused combination we shall exploit is 
\begin{equation}\label{eq:GammaToForms}
\begin{split}
    \epsilon_{AB}\left(\QDirac^{[-m]}C\Gamma^{A}\Gamma^{\pm}\QDirac^{[m]}\right)\dBasis^{B}  =\sum_{\substack{k=0\\r-|A|-k\ {\rm even/odd}}}^{\rank-|A|}(-1)^{\frac{k(k-1)}{2}}\QPS^{[m]}_{(\rank-|A|-k)} \wedge \QPS^{[-m]}_{(k)}\,.
\end{split}    
\end{equation}
A particularly important case is to consider a compatibility condition following from  \eqref{eq:VectorSpinorRelations} and to rewrite it using  \eqref{eq:GammaToForms}:
\begin{equation}\label{eq:Consistency}
\begin{split}
    \QDirac^- C\Gamma^{\fullset/\{a,b\}}\Projp \QDirac^+ =\QDirac^- C\Gamma^{\fullset/\{a,b\}}\Projm \QDirac^+ \quad
    \implies \quad W(\QPS_{ab},\PSes) = W(\QPS_a,\QPS_b)\,.
\end{split}
\end{equation}

Relation \eqref{eq:Consistency} together with \eqref{eq:PureSpinorCondition} allow rewriting  combination \eqref{eq:GammaToForms} in a remarkable way. Consider  the following example: set $r-|A| = 4$ in \eqref{eq:GammaToForms} and choose the case $\Projp$, then the \rhs of \eqref{eq:GammaToForms}  becomes
\be
\begin{aligned}
    \QPS^{[m]}_{(4)}\PSes^{[-m]} - \QPS^{[m]}_{(2)}\wedge \QPS^{[-m]}_{(2)} + \PSes^{[m]}\wedge \QPS^{[-m]}_{(4)}
    &= \frac{\PSes^{[m]}\PSes^{[-m]}}{2}\left(\left(\frac{\QPS_{(2)}}{\PSes}\right)^{[m]}-\left(\frac{\QPS_{(2)}}{\PSes}\right)^{[-m]}\right)^2\\
    &= \begin{cases}
        0 &  m=0,1,2\,, \\
        \frac{\QPS^{[3]}_{(1)}\wedge\QPS^{[1]}_{(1)} \wedge \QPS^{[-1]}_{(1)}\wedge \QPS^{[-3]}_{(1)}}{\PSes^{[2]_{\sh}}} & m=3\,,
        \end{cases}             
\end{aligned}
\ee
where we used \eqref{eq:Consistency} in the form notation: $\QPS^+_{(2)}\PSes^--\QPS^-_{(2)}\PSes^+ = \QPS^+_{(1)} \wedge \QPS^-_{(1)}$.

In general, using the same method, one finds
\be\label{eq:GammaAsForms}
\begin{aligned}
    \epsilon_{AB}\QDirac^{[-m]} C\Gamma^{A} \Projpm \QDirac^{[m]} \dBasis^{B}=
    \begin{cases}
        0\,, \quad m=0,1,\dots,\rank-2-|A| \\
       \frac{\QPS_{(1)}^{[\rank-|A|-1]}\wedge \QPS_{(1)}^{[\rank-|A|-3]}\dots \wedge \QPS_{(1)}^{[-\rank+|A|+1]}}{\PSes^{[\rank-2-|A|]_{\sh}}}\,,
        \quad
        m=\rank-1-|A|. 
        \end{cases}
\end{aligned}
\ee
The last relation suggests introducing a $\gl_\rank$ Q-system. Indeed, recall that by such a system we mean a collection of functions $\glQ_A$ that satisfy QQ-relations
\be
\label{eq:QQoriginal}
W(\glQ_{Aa},\glQ_{Ab})=\glQ_{Aab}\,\glQ_{A}\,.
\ee
These relations are solved in terms of single-index Q-functions and $\glQ_{\es}$: $\glQ_{A}=\glQ_{\es}^{-[|A|-1]_{\sh}}W(\glQ_{a_1},\glQ_{a_2},\ldots,\glQ_{a_{|A|}})$, and an analogous Wronskian determinant is what appears on the \rhs of \eqref{eq:GammaAsForms}. Note also that Hodge-dual functions $\glQ^A=\epsilon^{AB}\glQ_B$ also form a Q-system. The notations and main features of $\gl_\rank$ Q-systems suitable for arbitrary rank including their supersymmetric generalisation are due to Tsuboi \cite{Tsuboi:2009ud}, for a review and further details see e.g. \cite{Kazakov:2015efa}.

We shall label elements of the relevant for us $\gl_\rank$ Q-system as $\Qq_A$ and $\Qq^A$ defined as follows
\be\label{eq:QPsiMap}
\begin{aligned}
    &\Qq_{a}=\Pso\,\QPS_{a}\,, 
    \quad
    \Qq_{A} = W(\Qq_{a_1},\Qq_{a_2},\dots,\Qq_{a_{|A|}})\,, 
    \\
    &\Qq_{\emptyset}=1\,,\quad \Qq^{A} = \epsilon^{AB}\Qq_{B}\,,
\end{aligned}
\ee
where $\Phi\equiv \frac{1}{\PSes}$. We can think about $\Phi$ as the quantity used to rescale the projective coordinates $\Psi\to \Phi\, \Psi$ such that $\PSes\to 1$. Note that \eqref{eq:DynkinDiagramQQ} implies a very concrete normalisation of $\Psi$, \cf \eqref{eq:Quant}, and the rescaling with $\Pso$ does not preserve it. Hence we use the new notation $\Qq_a$ for the rescaling of $\Psi_a$, note also that $\Qq_{A}$ is defined through Wronskian determinant and hence {\it is not} a rescaling of $\Psi_{A}$ for $|A|>1$. Further benefits of this alternative normalisation are: First, the character limit of $\Qq_{a}$ coincides with the one of the genuine Q-functions of a $\gl_{\rank}$ spin chain; Second, another useful combination $\Psi_{ab}\to \Psi_{ab}/\Psi_{\es}$ is formed. We encountered it several times already, for instance in \eqref{eq:exp}, and shall give it a special labelling as well: $\mu_{ab}\equiv \Psi_{ab}/\Psi_{\es}$. To avoid possible confusion, we point out that all  relevant rescaled quantities received a new labelling and so functions $\Psi_{A}$ remain in their original normalisation in what follows.

\medskip
\noindent
Applying \eqref{eq:GammaAsForms} to \eqref{eq:VectorSpinorRelations} with $I=A$, we obtain
\begin{align}\label{eq:VectorUpperFromSpinor}
   &\hQV^{A} = \frac{{\hQq^{A}}}{ \Pso^{[r-1-|A|]}\Pso^{[-r+1+|A|]}}\,.
\end{align}

We remind that $A=a_1\ldots a_{|A|}$ is a multi-index with all entries being positive, and so \eqref{eq:VectorUpperFromSpinor} does not produce all components of $V^I$, in particular not all $V^i$. Let us now compute $V_{a} = V^{-a}$. To this end we notice that pure spinor conditions \eqref{eq:PureSpinorCondition} can be rewritten as
\begin{subequations}
\label{eq:PSDifferential}
\begin{align}
\label{eq:PSDifferentialEven}
    &(\partial_a - \mu_{ab}\theta^{b})\Gamma^+\Psi = 0\,,\\
\label{eq:PSDifferentialOdd}
    &(\partial_a- \mu_{ab}\theta^b)\Gamma^{-}\Psi = \Qq_a\,\Gamma^+ \Psi\,.
\end{align}
\end{subequations}
Now replacing Grassmann-variables and their derivatives with $\Gamma$-matrices as in \eqref{eq:GammaAsOperators} we can find $V_{a}$ as follows
\begin{equation}\label{eq:LowerVPsi}
    V_{a} = \Psi^{[-r+2]}C\Gamma_{a}\Gamma^\pm\Psi^{[r-2]}  =  \mu_{ab}^{[r-2]}\Psi^{[-r+2]}C\Gamma^{b}\Gamma^\pm\Psi^{[r-2]} = \mu_{ab}^{[r-2]}\, V^{b}\,,
\end{equation}
notice that for the choice of $\Gamma^-$ in $\Gamma^{\pm}$ the derivation additionally relies on \eqref{eq:ProjectionRelations} for $I=\es$.

The overall shift of $\mu_{ab}$ in \eqref{eq:LowerVPsi} is not uniquely fixed. To see this we use that \eqref{eq:Consistency}, conveniently written as
\begin{equation}   
\label{eq:shiftmutrick}
    \mu_{ab}^{[m]} = \mu_{ab}^{[m\pm 2]} \mp W(\Qq_{a},\Qq_{b})^{[m\pm 1]}\,,
\end{equation}
allows us to shift $\mu_{ab}$ at the cost of introducing terms involving $\Qq^{[m]}_{a},m\in\{-r+2,\dots,r-2\}$. However, when the shift of $\mu_{ab}$ is restricted to $[-r+2,r-2]$, these additional terms vanish due to
\be\label{eq:QuantisationSUN}
    \hQq_{b}^{[m]}\hQq^{b} = \begin{cases}
                           0\,,\quad m=-r+2,-r+4,\dots,r-2\,, \\
                            (\pm 1)^{\rank-1}\hQq_{\fullset}^{\pm}\,,\quad m=\pm\rank\,.
                           \end{cases}\,,
\ee
which is a direct consequence of definitions \eqref{eq:QPsiMap}.

Let us summarise the resulting expressions including also their inverse:
\begin{subequations}
\label{eq:VectorLowerFromSpinor}
\begin{align}
\label{eq:VectorLowerFromSpinora} 
   \qV_{a} &= \mu_{ab}^{[m]}V^b\,, 
   &
   \mu_{ab} &\equiv \Pso\,\QPS_{ab}\,,
  \\
\label{eq:VectorLowerFromSpinorInverse}
  \qV^{a} &=
   (\mu^{ab})^{[m]}\qV_{b} \,,
   &
   \mu^{ab} &\equiv -\frac{1}{\QPS_{\fullset}}\,\QPS^{ab}\,,
\end{align}
\end{subequations}
for $m\in \{-\rank+2,-\rank+4,\dots,\rank-2\}$. 
For even $\rank$, $\mu^{ab}$ is the inverse of $\mu_{ab}$. For odd $\rank$, the determinant of $\mu_{ab}$ vanishes and consistency of \eqref{eq:VectorLowerFromSpinor} is less obvious. In this case we have to use first $\mu^{ab}\mu_{bc} = \delta^{a}{}_{c} - \frac{\Psi^{a}\Psi_{c}}{\Psi_{\es}\Psi_{\fullset}}$ and then once again \eqref{eq:QuantisationSUN}.

Finally, we also want to express $V^{I}$ with any multi-index $I$ using only spinors. To this end represent $V^I$ in the form $V_{A}{}^{B}$ by lowering the negative indices according to \eqref{eq:ama}. Then we systematically use local pure spinor properties \eqref{eq:PSDifferential} and fused pure spinor properties \eqref{eq:ProjectionRelations}  and \eqref{eq:VectorSpinorRelations}. As an example let us calculate $V_{a}{}^{b}$:
\begin{equation}
    \begin{split}
        V_{a}{}^{b} &=  \Psi^{[-r+3]}C\Gamma_{a}{}^{b}\Gamma^- \Psi^{[r-3]} \\
        &=-\Psi^{[-r+3]}C \Gamma^{b}\Gamma_a \Gamma^- \Psi^{[r-3]} + \frac{\delta^{a}_{b}}{2} \underbrace{\Psi^{[-r+3]}C\Gamma^-\Psi^{[r-3]}}_{=0} \\
        &=\mu_{ac}^{[r-3]} \Psi^{[r-3]}C\Gamma^{cb}\Gamma^-\Psi^{[r-3]} - \Qq_a^{[r-3]}\underbrace{\Psi^{[-r+3]}C\Gamma^b\Gamma^+\Psi^{[r-3]}}_{=0} \\
        &=\mu^{[r-3]}_{ac}\frac{\bP^{cb}}{\Phi^{[r-3]}\Phi^{[-r+3]}}\,.
    \end{split}
\end{equation}
This computation generalises directly to any number of indices, we have
\begin{equation}
    V_{A}{}^{B} = \mu^{[r-1-|A|-|B|]}_{AA'}\Psi^{[-r+1+|A|+|B|]}C\Gamma^{A'B}\Gamma^\pm\Psi^{[r-1-|A|-|B|]} + \dots\,,
\end{equation}
where $\mu_{AA'}$ is the determinant $|\mu_{aa'}|_{\substack{a\in A\\a'\in A'}}$ and all suppressed terms will vanish. In summary, we arrive at the compact expression 
\begin{align}
\label{eq:VQPsi}
\boxed{
\hQV_{A}{}^{B}=\frac{\mu_{AA'}^{[m]}\,\hQq^{A'B}}{\Pso^{[r-1-|A|-|B|]}\Pso^{[-r+1+|A|+|B|]}}
}\,,
\end{align}
where $m\in \{-\rank+1+|A|+|B|,\dots,\rank-1-|A|-|B|\}$, the admissible range of $m$ follows from the same logic that lead to \eqref{eq:VectorLowerFromSpinor}. The formula is valid for $|A|+|B|\leq \rank-1$, \ie for all $V^{I}$ defining isotropic hyperplanes. Recall that~\cite{Ekhammar:2020enr}~self- and anti-self-dual hyperplanes are not given by $V^{I}$, $|I|=\rank$, but by $\Psi C\Gamma^{(\rank)}\Gamma^{\pm}\Psi$---a combination which, in contrast to \eqref{eq:VectorSpinorRelations}, has $\Psi$ without shifts.

\medskip{}
\noindent
We also notice that $V^{\emptyset}=1$ which means that $\hQq^{\emptyset}=\hQq_{\fullset}=W(\hQq_1,\ldots,\hQq_\rank)=\Pso^{[r-1]}\Pso^{[-r+1]}$, \cf \eqref{eq:VectorUpperFromSpinor}.
\subsection{Pure spinor Q-system and Wronskian Bethe equations}
We came to conclusion that relations \eqref{eq:VQPsi} and \eqref{eq:PureSpinorCondition} compute all Q-functions of the extended Q-system directly from $\Psi_{\emptyset},\Psi_{a},\Psi_{ab}$, that is from $1+\rank+{\rank\choose 2}$ functions. We however also know that the functional freedom in a Q-system is equal to the rank of the Lie algebra and thus there should be $1+{\rank\choose 2}$ relations between functions of the triple. These relations can be indeed found in the above-introduced zoo of formulae. They play the central role for our study:
\begin{center}
\begin{tcolorbox}[title=Wronskian Bethe equations,top=-2.5mm,width=80mm]
\begin{subequations}\label{eq:SpinorQSystem}
    \begin{align}
        \label{eq:DeterminantCondition}
        &W(\QPS_1,\ldots,\QPS_\rank) = \PSes^{[\rank-2]_{\sh}}\,,\\
        \label{eq:U1Conditions}
      &W(\QPS_a,\QPS_b)=W(\QPS_{ab},\PSes)\,.
    \end{align}
\end{subequations}
\end{tcolorbox}
\end{center}
\noindent Here the first line is yet another way to write \eqref{eq:VectorUpperFromSpinor} for $A=\emptyset$, and the second line is \eqref{eq:Consistency}. 

 At least formally, for any choice of $\QPS_a$, \eqref{eq:SpinorQSystem} can be always solved for $\PSes$ and $\QPS_{ab}$ and so one may wonder how the above equations contain any interesting information. Without imposing requirements on  analytic properties of Q-functions, these equations are indeed essentially non-constraining. But our attention at this stage is different: that they imply all other relations of the extended Q-system, so the functional freedom is indeed $\rank$ independent functions as expected. This is the subject of Sections~\ref{sec:DerivationOfExtendedQsystem} and \ref{sec:Uniquenss}. Eventually, when one imposes analytic restrictions on Q-functions to describe the desired physics, as in Section~\ref{sec:RationalSpinChains}, equations \eqref{eq:SpinorQSystem} become far from trivial.

\medskip
\noindent
It is interesting to juxtapose \eqref{eq:U1Conditions} with QQ-relations \eqref{eq:QQoriginal} of $\gl_{\rank}$ Q-system. To elaborate even further, we notice a more general relation that holds as well
\begin{center}
\begin{tcolorbox}[title=Pure spinor Q-system,top=-2.5mm,width=80mm]
\be
\label{eq:PsiPsi}
W(\Psi_{Aa},\Psi_{Ab})=W(\Psi_{Aab},\Psi_{A})\,.
\ee
\end{tcolorbox}
\end{center}
This relation is in fact rewriting of $W(\hQ_{(r-1)}^{\weyl(1)},\hQ_{(r-1)}^{\weyl(2)})=W(\hQ_{(r)}^{\weyl(1)},\hQ_{(r)}^{\weyl(2)})$, a consequence of \eqref{eq:DynkinDiagramQQ}, using the exterior algebra parameterisation \cf Appendix~\ref{sec:Weyl}.

\medskip
\noindent
We shall call {\it pure spinor Q-system} a collection of functions $\Psi_{A}$ that satisfy \eqref{eq:PsiPsi}. Curiously, we do not need to even require Cartan formulae \eqref{eq:PureSpinorCondition} as they arrive naturally as a consequence of \eqref{eq:PsiPsi}:
\begin{lemma}[Solution of pure spinor Q-system]
\label{leamma:purespinor}
If $\Psi_{\emptyset},\Psi_{a},\Psi_{ab}$ satisfy \eqref{eq:U1Conditions} and $\Psi_{\emptyset}\neq 0$ then $\Psi_A$ computed via \eqref{eq:PureSpinorCondition} solve \eqref{eq:PsiPsi}. Moreover, this solution is unique provided combinations $W(\Psi_{Aa},\Psi_{Ab})$ are not zero.
\end{lemma}
This lemma is the analog of the fact in $\gl_{\rank}$ case that $\glQ_{A}=W(\glQ_{a_1},\glQ_{a_2},\ldots,\glQ_{a_{|A|}})$ unambiguously solves \eqref{eq:QQoriginal} provided $\glQ_{A}\neq 0$ (for simplicity, we have set $\glQ_\es=1$). The proof in $\gl_{\rank}$ case relies on using \Plucker identities or equivalently relations between minors of Wronski matrix (a matrix whose entries are $\glQ_{a}^{[2n]}$), see \eg\cite{Kazakov:2015efa}. The proof in the pure spinor case is similar in spirit but it uses additionally \eqref{eq:U1Conditions} and overall is considerably longer, we provide it in Appendix~\ref{app:techproofs}. If using minors in $\gl_{\rank}$ case is a determinant-based computation, for pure spinor system we use techniques typical for operations with Pfaffians.

\medskip
\noindent
To better understand the role of \eqref{eq:DeterminantCondition}, we note that relations \eqref{eq:PsiPsi} are projective in the sense that rescaling $\Psi_{A}=g\,\Psi_{A}$ with arbitrary function $g(u)$ is their symmetry. An analogous rescaling in the $\gl_{\rank}$ case is $\glQ_A\to g^{[\,|A|\,]_\sh}\, \glQ_A$. Fixing normalisation is necessary and note that it removes one functional degree of freedom, then one speaks about $\sl_\rank$ system instead of $\gl_\rank$. A way to go is to set $W(\glQ_1,\ldots,\glQ_r)=1$. By itself the requirement $W(\glQ_1,\ldots,\glQ_r)=1$ is of no particular meaning but combined with concrete analytic Ansatz on Q-functions it becomes a way to describe Bethe algebra, it is in fact the only relation needed for a spin chain with sites in vector representations \cite{MTV},\cite{Chernyak:2020lgw} replacing the need of using nested Bethe equations; for this reason we refer to $W(\glQ_1,\ldots,\glQ_r)=1$ as $\sl_{\rank}$ Wronskian Bethe equations\footnote{A pedantic way would be to call Wronskian Bethe equations the equations that come in the consequence of an analytic Ansatz on Q-functions applied to $W(\glQ_1,\ldots,\glQ_r)=1$. We found it convenient to lift the name directly up to $W(\glQ_1,\ldots,\glQ_r)=1$ as it is a suitable place holder for a variety of different Ansatze to try.}.

Setting the normalisation for $\Psi_A$ is done to get the agreement with the extended Q-system. The latter was introduced with fixed normalisation as is clearly visible in \eqref{eq:Quant}. The requirement of compatibility in normalisations is realised by \eqref{eq:DeterminantCondition} and so we may consider it as the equivalent of  $W(\glQ_1,\ldots,\glQ_r)=1$, a proposal that was already made in Section~7 of \cite{Ferrando:2020vzk}, equation~(7.2) there. However, we observe that it is not enough to use only \eqref{eq:DeterminantCondition} and the model-specific analytic input to fix the Bethe algebra. One of the messages of our work,  Section~\ref{sec:RationalSpinChains}, is that \eqref{eq:DeterminantCondition} and \eqref{eq:U1Conditions} together form a system that properly describes the Bethe algebra. It then makes sense to collectively refer to \eqref{eq:SpinorQSystem} as $\so_{2\rank}$ Wronskian Bethe equations.

\medskip
\noindent
Instead of $\Psi_{\emptyset},\Psi_{a},\Psi_{ab}$, we can also select $\Phi,\Qq_a,\mu_{ab}$  as our main functions thus forming
\begin{center}
\begin{tcolorbox}[title=The $\mathfrak{so}_{2r}$ ${\bf P}\mu$-system,top=-2.5mm,width=80mm]
\begin{subequations}\label{eq:PmuSystem}
    \begin{align}
        \label{eq:PmuSystema}
        &W({\bf P}_1,\ldots,{\bf P}_\rank) = \Phi^{[r-1]}\Phi^{[-r+1]}\,,\\
        \label{eq:PmuSystemb}
       &W({\bf P}_a,{\bf P}_b)=\mu_{ab}^+-\mu_{ab}^-\,
    \end{align}
\end{subequations}
\end{tcolorbox}
\end{center}
\noindent which shares similarities with the ${\bf P}\mu$-system of AdS$_5$/CFT$_4$ \citep{Gromov:2013pga,*Gromov:2014caa}. One important difference is that raising/lowering indices using the `symplectic structure' $\mu_{ab}$ does not close the system on itself but rather expands it introducing the second half of vector Q-functions $V^i$ according to \eqref{eq:VectorLowerFromSpinor}.

\subsection{Deriving the extended Q-system}
\label{sec:DerivationOfExtendedQsystem}
The triple $\QPS_{\emptyset},\QPS_{a},\QPS_{ab}$ allows computing other Q-functions by sheer combinatorics, but even more remarkable is that all relations of the extended Q-system can be obtained as a consequence of \eqref{eq:SpinorQSystem}:
\begin{theorem}
\label{thm:derivingextendedQsystem}
If relations \eqref{eq:SpinorQSystem} for $\Psi_\es,\Psi_a,\Psi_{ab}$ are satisfied then components of Q-vectors $Q_{(a)}$ computed using \eqref{eq:VQPsi} and \eqref{eq:PureSpinorCondition}, $V_{(r-1)}$ computed using \eqref{eq:VQPsi}, and $V_{(r)}$ computed using \eqref{eq:TensorQ}~\footnote{Note that \eqref{eq:TensorQ} computes $V_{(r)}$ from $V_{(1)}\equiv Q_{(1)}$, and it is meant that $Q_{(1)}$ is the one computed from \eqref{eq:VQPsi}.} satisfy all the relations listed in 
Section~\ref{sec:FusedFlags}.
\end{theorem}
\begin{proof} Two key features associated with $\mu_{ab}$ participate in the majority of the arguments. The first feature is the possibility to change the value of shift in $\mu_{ab}^{[m]}$ in an admissible range in various formulae, this is based on  the trick of using \eqref{eq:shiftmutrick} together with \eqref{eq:QuantisationSUN}. Note that \eqref{eq:shiftmutrick} is just a rewriting of \eqref{eq:U1Conditions} while \eqref{eq:QuantisationSUN} is a direct consequence of how $\bP^a$ are defined, so using the shift trick relies only on the assumptions of the theorem and hence is allowed in the proof. 

This first feature is in particular used to get an admissible range of shifts in \eqref{eq:VQPsi} which  allows in turn deriving Wronskian relation \eqref{eq:TensorQ} for $V^{I}$, $|I|\leq r-1$. Indeed, assume $\Pso=1$ for simplicity, then
\begin{align}
V_{A}{}^{B}=V_{a_1\ldots a_k}{}^{b_1\ldots b_l}
=\mu_{a_1a_1'}^{[m]}\ldots \mu_{a_ka_k'}^{[m]}W(V^{a_1'},\ldots,V^{a_k'},V^{b_1},\ldots,V^{b_l})&
\nonumber\\
=W(\mu^{[\epsilon]}_{a_1a_1'}V^{a_1'},\ldots,\mu^{[\epsilon]}_{a_ka_k'}V^{a_k'},V^{b_1},\ldots,V^{b_l})&=
W(V_{a_1},\ldots,V_{a_k},V^{b_1},\ldots,V^{b_l})\,,
\end{align}
where $\epsilon=0$ for even rank and $\epsilon=\pm 1$ for odd rank. What we used is $\mu_{aa'}^{[m]}(V^{a'})^{[s]}=\mu_{aa'}^{[s+\epsilon]}(V^{a'})^{[s]}$, since for all $s$ appearing in Wronskian determinant $m=s$ happens to be in the admissible range, this allowed us to pull $\mu_{a_ia_i'}$ inside the determinant. Checking that the argument works in the same way when $\Phi\neq 1$ is straightforward, this however relies additionally on the normalisation $V^{\es}=\frac{\bP^{\es}}{\Pso^{[r-1]}\Pso^{[-r+1]}}=1$ which is \eqref{eq:DeterminantCondition} rewritten in the form \eqref{eq:PmuSystema}. 

As discussed after \eqref{eq:QQoriginal}, having \eqref{eq:TensorQ} also implies $W(V^{Ii},V^{Ij})=V^{Iij}V^{I}$ which are precisely QQ-relations \eqref{eq:DynkinDiagramQQ} for $a=1,\ldots,r-3$, where $a=|I|+1$.

\medskip
\noindent
The second key feature involving $\mu_{ab}$ is using \eqref{eq:PSDifferential}---a way to present Cartan formulae \eqref{eq:PureSpinorCondition}---to effectively raise/lower indices. 
The first example of its usage is a derivation of the fused pure spinor relation \eqref{eq:ProjectionRelations}, including  \eqref{eq:PureSpinor} as a particular case. If one takes $I=A$, \ie all positive indices, in $\Gamma^I$ of \eqref{eq:ProjectionRelations} then the fused pure spinor relation is precisely what is stated in the first line of \eqref{eq:GammaAsForms}; we notice that the latter is derived by using solely \eqref{eq:U1Conditions} and Cartan formulae and so it can be used in this proof as well. For arbitrary $I$, we employ \eqref{eq:PSDifferential} to replace any occurrence of $\Gamma_a$ with $\mu_{ab}\Gamma^b$ which reduces the question to the already understood case of $\Gamma^A$. As \eqref{eq:PSDifferentialOdd} has r.h.s., we note that $\bP_a$ is a scalar prefactor  and so the question of vanishing of the \lhs of \eqref{eq:ProjectionRelations} involving $\Gamma^{I}$ is replaced with the one involving  $\Gamma^{I'}$ with $|I'|=|I|-1$. 

Derivation of fused Fierz relations \eqref{eq:VectorSpinorRelationsa} is the reverse of how \eqref{eq:VQPsi} were obtained. Indeed, the equivalence of \eqref{eq:VectorSpinorRelationsa} and \eqref{eq:VQPsi} was demonstrated through a set of equalities which are based on features already established above in the logic of the current proof. 

Fused Fierz relations \eqref{eq:VectorSpinorRelationsa} for $|I|=r-2$ are  QQ-relations \eqref{eq:DynkinDiagramQQ} for $a=r-1$ with the choice of $\Gamma^-$ in \eqref{eq:VectorSpinorRelationsa}, and for $a=r$ with the choice of $\Gamma^+$.  QQ-relations \eqref{eq:DynkinDiagramQQ} for $a=r-2$ is $W(V^{Ii},V^{Ij})=V^{Iij}V^{I}$ for $|I|=r-3$ with $V^{Iij}$  replaced by the bi-linear combination of spinors using \eqref{eq:VectorSpinorRelationsa} for $|I|=r-1$. Note that at this stage we have accomplished the derivation of all QQ-relations \eqref{eq:DynkinDiagramQQ}.

\medskip
\noindent
Both key features involving $\mu_{ab}$ are exploited in the derivation of fused orthogonality properties of vectors \eqref{eq:fusedortho}:
\begin{align}
\label{eq:projection}
(V^{i})^{[m]}(V_{i})^{[-m]}=(V^{a})^{[m]}(V_{a})^{[-m]}+(V_{a})^{[m]}(V^{a})^{[-m]}
=
(V^{a})^{[m]}(\mu^{[m]}_{ab}-\mu^{[-m]}_{ab})(V^{b})^{[-m]} = 0\,,\quad |m|\leq r-2\,, 
\end{align}
where we used a possibility to raise index $V_a=\mu_{ab}V^b$, antisymmetry of $\mu_{ab}$ and, at the last step, the possibility to replace $\mu_{ab}^{[-m]}$ with $\mu_{ab}^{[m]}$ using the trick based on \eqref{eq:shiftmutrick} and \eqref{eq:QuantisationSUN}.

If  $m=\rank-1$, one cannot use the shifting trick  to set $m=-m$ in \eqref{eq:projection} because the contribution from $W(\Qq_a,\Qq_b)$ survives on one occasion when using \eqref{eq:shiftmutrick}. Then we invoke the second line of \eqref{eq:QuantisationSUN} and ultimately \eqref{eq:DeterminantCondition} to conclude the normalisation property
\begin{align}
\label{eq:V}
    (\hqV^{i})^{[\rank-1]}\hqV_{i}^{[-\rank+1]}=(\hqV^a)^{[-\rank+1]}\Qq_{ab}(\hqV^{b})^{[\rank-1]}  
  =(V^{a[-r]}\Qq_a)^+(\Qq_b V^{b[r]})^- = (-1)^{r-1}\,.
\end{align}
Now we use \eqref{eq:projection} and \eqref{eq:V} to derive \eqref{eq:QuantVector}: \mbox{$
    (\hqV^{I})^{[\rank-1]}\hqV_{I}^{[-\rank+1]} \!\!=((\hqV^{i})^{[\rank-1]}\hqV_{i}^{[-\rank+1]})^{[|I|]_{\sh}} \!\! =\! (-1)^{|I|(r-1)}\,.$} 

The normalisation property of spinors \eqref{eq:QuantSpinor} follows from the second line of \eqref{eq:GammaAsForms} for $A=\emptyset$, and \eqref{eq:DeterminantCondition}. 

The last relation to demonstrate is \eqref{eq:VectorSpinorRelationsb}. Its derivation uses a combination of the already presented techniques, we give it in Appendix~\ref{app:techproofs}.
\end{proof}

\subsection{\label{sec:Uniquenss}Existence and uniqueness}
To conclude this section, we use the gained knowledge to build the explicit transition between $Q_{(a)}^1$, $a=1,\ldots,\rank$, and the pure spinor Q-system. We recall that using $Q_{(a)}^1$ is a historically better-known way to store information about the Bethe algebra: in the case of spin chains and up to a normalisation, $Q_{(a)}^1$ are the eigenvalues of transfer matrices in prefundamental representations \cite{Frenkel_2015}; q-character expressions are written as Laurent polynomials in  $\mathcal{Y}_a={Q_{(a)}^{1}}/{Q_{(a)}^{1\,[-2]}}$; zeros of $Q_{(a)}^1$ satisfy nested Bethe Ansatz equations, see \eg \cite{Reshetikhin:1983vw,Kuniba:1994na,Frenkel:1998ojj} and also Section~\ref{sec:5a}.

We already showed that all the Q-functions $Q_{(a)}^{\mathfrak{i}}$ and hence $Q_{(a)}^1$ are computed unambiguously from $\Psi_{\emptyset},\Psi_a,\Psi_{ab}$. It is the opposite statement which we shall demonstrate now:

\begin{theorem}
For any general position functions $Q_{(a)}^1$, $a=1,\ldots r$, there is the unique up to a symmetry set of functions $\Psi_{\emptyset},\Psi_a,\Psi_{ab}$ such that the relations of the extended Q-system are satisfied.
\end{theorem}

Before doing the proof, we explain what `up to a symmetry' means in practice. All relations of the extended Q-system are invariant under transformations $Q_{(a)}\to G\cdot Q_{(a)}$, where $G$ is a $\mathsf{Spin}(2\rank)-$valued function of the spectral parameter that satisfies $G^+=G^-$, and it acts on $Q_{(a)}$ as on the vector in $a$-th fundamental representation. As $Q_{(a)}^1$ are fixed, the only possible symmetries are transformations from ${\rm exp}(\mathfrak{n}_-)$, where $\mathfrak{n}_-$ is generated by the lowering operators of $\mathfrak{so}_{2\rank}$. Explicitly, these are $\Psi_a\to\Psi_a+\sum\limits_{b>a}g_{ab}\Psi_b$ and $\Psi_{ab}\to\Psi_{ab}+{\tilde g}_{ab}\Psi_\es$ for some $g_{ab}^+=g_{ab}^-$, $\tilde g_{ab}^+=\tilde g_{ab}^-$, \cf \eqref{eq:419}. These transformations is the allowed ambiguity in the derivation of the triple $\Psi_\es,\Psi_a,\Psi_{ab}$.

\begin{proof}
In the notations of this paper, $Q_{(a)}^1=V^{12\ldots a}$ for $a\leq \rank-2$, $Q_{(r-1)}^1=\Psi_{r}$, $Q_{(r)}^1=\Psi_{\es}$. Using definitions of $\Phi$ and $\bP_A$, as well as \eqref{eq:VectorUpperFromSpinor}, we see that this input is equivalent to knowing $\bP_{\es}=1$, $\bP_{\rank},\bP_{\rank-1,\rank},\ldots, \bP_{\fs}\equiv \bP_{1\ldots\rank}$. This set of $\bP$'s plays the role of the set of $\glQ_{1\ldots a}\equiv \glQ_{(a)}^{1}$ for $\gl_r$ Q-system. From it, getting all the components $\glQ_{a}$ of $\mathcal{Q}_{(1)}$ (in our notations these would be $\bP_a$) is a well-established procedure which we repeat here for clarity: Let $a\ldots b$ mean a sequence of all integers from $a$ to $b$ (with $b\geq a>0$). We recursively use QQ-relations
\be
\label{eq:639}
W(Q,\bar Q)=J\,,
\ee
where $Q=\bP_{a+1\ldots b}$, $\bar Q=\bP_{a\ldots b-1}$, $J=-\bP_{a+1\ldots b-1}\,\bP_{a\ldots b}$,
to compute $\bP_{a\ldots b-1}$ from the known $\bP_{a\ldots b}$, $\bP_{a+1\ldots b}$, $\bP_{a+1\ldots b-1}$. The input for the recursion is the known from $Q_{(a)}^1$ functions $\bP_{A}$ that were listed above. The procedure terminates with the computation of single-indexed $\bP_a$, $1\leq a\leq \rank$.

We get $\Psi_a$ from $\bP_a$ by a simple rescaling.

Solution of \eqref{eq:639} for $\bar Q$ is not unique but defined up to a replacement ${\bar Q}\to{\bar Q}+h\,Q$, where $h^+=h^-$; and any two solutions are related in this way. To compute $\bP_{a}$, we need to perform $\rank\choose 2$ steps of the recursion, and the possible solution ambiguities account precisely for the symmetry transformations $\Psi_a\to\Psi_a+\sum g_{ab}\Psi_b$.

Once $\Psi_a$ are known, $\Psi_{ab}$ are obtained as solutions to \eqref{eq:U1Conditions}. The latter should be viewed as of type \eqref{eq:639} with $\Psi_{ab}$ playing the role of $\bar Q$, so $\Psi_{ab}$ are fixed up to the symmetry ambiguity $\Psi_{ab}\to\Psi_{ab}+{\tilde g}_{ab}\Psi_\es$. 

Only a subset of relations of the extended Q-system was used in the computation of $\Psi_{\es},\Psi_{a},\Psi_{ab}$. So let us assume that a different computation algorithm using different relations was suggested. Let us take the triple $\Psi_{\es},\Psi_{a},\Psi_{ab}$ obtained in this alternative way and compute all other Q-functions from it. Owing to Theorem~\ref{thm:derivingextendedQsystem}, all the relations of the extended Q-system shall be satisfied  by these Q-functions. Focusing on the relations in the alternative algorithm we confirm that $Q_{(a)}^1$ computed from the triple coincide with the ones we started from. But also the relations that were used to derive $\Psi_{\es},\Psi_{a},\Psi_{ab}$ following guidelines of this proof shall be valid as well. So two different computation approaches must produce the same up to the symmetry triples. This argumentation assumes that the suggested algorithms can be indeed performed producing at every their step valid functions starting from the chosen concrete functions $Q_{(a)}^1$, and thus we assume general position initial data. More comments on this matter are given later in this section.

In conclusion, the values of $\Psi_{\es},\Psi_{a},\Psi_{ab}$, up to the symmetry, do not depend on  the choice of a procedure to compute them.
\end{proof}
\begin{corollary}
For any general position functions $Q_{(a)}^1$, $a=1,\ldots r$, it is always possible to define Q-vectors $Q_{(a)}=\sum Q_{(a)}^{\mathfrak{i}}\basvec_{\mathfrak{i}}^{(a)}$ such that the relations of the extended Q-system hold. This construction is unique up to the symmetry.\qed
\end{corollary}
There are various ways to fix the symmetry, for instance by picking a concrete solution of the QQ-system on the Weyl orbit \eqref{eq:DynkinDiagramQQ}. Then the above Corollary may be stated differently:
\begin{corollary}
For any general position functions $Q_{(a)}^{\weyl(1)}$ that solve \eqref{eq:DynkinDiagramQQ}, where $a=1,\ldots r$ and $\weyl$ are elements of Weyl group, there is a unique way to define Q-vectors $Q_{(a)}=\sum Q_{(a)}^{\mathfrak{i}}\basvec_{\mathfrak{i}}^{(a)}$ such that the relations of the extended Q-system hold.\qed
\end{corollary}
We refer to \cite{Ekhammar:2020enr} for the definition of Q-functions on the Weyl orbit $Q_{(a)}^{\weyl(1)}$ distilling here only what it means in practice: those are all components of $\Psi$, and such $V_{A}{}^{B}$ that $A\cap B=\emptyset$. Computing $Q_{(a)}^{\mathfrak{i}}$ outside of the Weyl orbit is easy, e.g. by \eqref{eq:TensorQ}, we however want to put the emphasis with the last corollary that all relations of the extended Q-system follow from \eqref{eq:DynkinDiagramQQ}. 

\medskip
\noindent
We see that the set of $Q_{(a)}^1$ and the set of $\Psi_{\es},\Psi_a,\Psi_{ab}$ contains the same information, but there is an important difference about how we extract it. Construction of all Q-functions from $\Psi_{\es},\Psi_a,\Psi_{ab}$ is straightforward as it requires only arithmetic operations on functions of the triple (also with appropriate shifts of the spectral parameter).  In contrast, doing the same from $Q_{(a)}^1$ requires solving relations of type \eqref{eq:639} which is more subtle. Of course, the difficulty of solving \eqref{eq:639} is not absent from the pure spinor Q-system, it is the difficultly of solving relations \eqref{eq:SpinorQSystem}.

Let us discuss in more detail how \eqref{eq:639}  is meant to be solved. Recall that so far we work in a general set up where no concrete physical model is assumed.  We therefore cannot require beforehand any analytic properties on Q-functions but only that their values are well-defined. In such generality, we can find $\bar Q$ as follows: decompose all values of the spectral parameter into $2\mathbb{Z}$-orbits, where $\mathbb{Z}$ is the shift group, \ie into the non-intersecting sets $\{u_0+\ii\,n|n\in\mathbb{Z}\}$. For each such set fix the initial value $\bar Q(u_0)$ at will, this precisely corresponds to the ambiguity ${\bar Q}\to{\bar Q}+h\,Q$ of the solution parameterised by $h$, and then find ${\bar Q}(u_0+\ii\,n)$ recursively by using \eqref{eq:639}. For instance, for $n>0$:
\be\label{eq:651}
\bar Q^{[2n]}=\left(\frac{\bar Q}{Q}-\sum_{k=0}^{n-1}\left(\frac{J}{Q^+Q^-}\right)^{[2k+1]}\right)Q^{[2n]}\,.
\ee
'General position' requirement used in the theorem above refers to the assumption that values of Q-functions are such that solutions for $\bar Q$ are still well-defined functions despite of appearance of Q-functions in the denominators. The same remark goes for usage of other formulae like \eqref{eq:PureSpinorCondition} and \eqref{eq:VQPsi}. In other words: If some function $Q$ or a collection of functions are used as an input for computing other Q-functions then such computation eventually reduces to computing a rational combination of a finite number of $Q^{[n]}$. Treating $Q^{[n]}$ with different $n$ as independent variables, the computations presented in this paper, including those used in the proofs, are valid on a Zariski-open subset of possible values for $Q^{[n]}$.

The above approach is sufficient for combinatorial reasoning related to the extended Q-system but one can do better if we assume mild analytic features of Q-functions. Namely, if there exist such $H(t)$ and contour $\gamma$ that the integrals
\be
\frac{J}{Q^+Q^-}=\int_{\gamma}e^{u\,t}H(t)dt\,,\quad \ \bar Q(u)=Q(u)\int_{\gamma}\frac{e^{u\,t}}{e^{-it/2}-e^{it/2}}\,H(t)dt\,,
\ee
are defined then the offered integral expression for $\bar Q$ is an analytic in $u$ function that solves \eqref{eq:639} for certain domain of the spectral parameter. Formuale like \eqref{eq:PureSpinorCondition} and \eqref{eq:VQPsi} are just rational combinations of Q-functions and they also naturally offer an analiticity domain for a Q-function they compute. It is in this sense we understand the existence theorem then: that Q-functions can be constructed within a certain domain of analyticity.

We emphasise that the analyticity domain is not just an artefact of the integral representation as there indeed can be singularities in $\bar Q$ reducing such a domain. For instance, any singularities of $\frac{J}{Q^+Q^-}$, in particular those coming from zeros of $Q$, typically result in a semi-infinte ladder of singularities in $\bar Q$. Although the integral representation disguises this feature, it becomes transparent from \eqref{eq:651}. Cancellation of these singularities {\it is not required} in the existence theorem and overall in the combinatorial study of Q-system in this paper. However, such a cancellation is often either a consequence or one of the explicit requirements of the analytic Ansatz on Q-functions which one imposes to precise the physical model.

\section{\label{sec:Tfunc} T-functions}
T-functions $T_{a,s}$ have the meaning of transfer matrices in the context of integrable spin chains with Kirillov-Reshetikhin module of highest weight $s\,\omega_a$ in the auxiliary space; certain spectral determinants in the context of ODE/IM correspondence; their ratios are Y-functions in the context of TBA.  They emerge within the extended Q-system as inner products
\begin{subequations}
\label{eq:Ts}
\begin{align}
\label{eq:Tvector}
   &T_{a,s} = (-1)^{a(\rank-1)}(V^I)^{[\rank-1+s]}V_{I}^{[-\rank+1-s]}\,,
   |I|=a\leq r-2\,,
   \\
\label{eq:Tspinor}
   &T_{r-1,s} = \QDirac^{[-\rank+1-s]} C\Projm\QDirac^{[\rank-1+s]}\,, \\
   \label{eq:Tspinor2}
   &T_{r,s} = \QDirac^{[-\rank+1-s]}C\Projp\QDirac^{[\rank-1+s]}\,,
\end{align}
\end{subequations}
where $s=0$ (the trivial representation case) corresponds to the normalisation condition \eqref{eq:Quant}. 
Expressions \eqref{eq:Ts} were obtained as solutions to Hirota equations in \cite{Ferrando:2020vzk},\cite{Ekhammar:2020enr} and they also, at least partially, are supported by operatorial constructions in \cite{Frassek:2020nki},\cite{Ferrando:2020vzk},\cite{Costello:2021zcl}.

T-functions are $\so_{2\rank}$-invariants~\footnote{with respect to the Langlands-dual action, the invariance with respect to the action of physical $\so_{2\rank}$ is not required in this section although it might hold as well. For the distinction between the two actions, see the discussion in the paragraph on page~\pageref{par:PLd}.} and we shall explore their decomposition in terms of $\gl_{\rank}$-invariants. In conclusion we shall get several relatively compact expressions. Results for $T_{a,s}$, $a\leq r-2$ and $T_{a,1}$, $a=r-1, r$ are present already in \cite{Ferrando:2020vzk} or earlier works modulo change of conventions, while expressions for fermionic $T_{a,s}$ ($a=r-1,r$) for $s>1$ are new.

For future convenience, we introduce  $\gl_{\rank}$-invariants using Weyl-type determinant formulae
\begin{align}
    \Tg_{\lambda}&=\frac{|\Qq_a^{[2\hat\lambda_b]}|}{\Qq_{\fullset}}
    \,,&
      \bTg_{\lambda}&=(-1)^{\frac{r(r-1)}2}\frac{|\Qq_a^{[-2\hat\lambda_b]}|}{\Qq_{\fullset}}\,,
\end{align}
where $\lambda=(\lambda_1,\ldots,\lambda_N)$ (typically an integer partition), and $\hat\lambda_a=\lambda_a-a+\frac{r+1}{2}$. Observe that $\bar T_{\lambda}$ is the complex conjugate of $T_{\lambda}$ when $\Qq_a$ are real. $\bTg_{\lambda}$ becomes $\gl_{\rank}$-character $\chi_{\lambda}$ in the case of the character solution of $\so_{2\rank}$ Q-system \cite{Ekhammar:2020enr} for which $\Qq_{a} \propto x_a^{\ii \spa}$.

If \eqref{eq:QPsiMap} were the Q-system of a $\gl_{\rank}$ spin chain, for instance the rational spin chain with monodromy matrix $M_{ab}(u)=\delta_{ab}+\ii\frac{E_{ab}}{u}+\ldots$ satisfying Yangian RTT relations $[M_{ab}(u),M_{cd}(v)]=\frac{\ii}{u-v}(M_{cb}(u)M_{ad}(v)-M_{cb}(v)M_{ad}(u))$, then  $\Tg_{\lambda}$ would correspond to the transfer matrix with irrep $\lambda$ in the auxiliary space. It is unclear whether this interpretation is realised in practice because we work with $\so_{2\rank}$ systems, but we use it as a useful analogy. In particular, certain transformations of Q-functions can be viewed as inherited from $\Yangian(\gl_{\rank})$ automorphisms:
\begin{equation}
\label{eq:auto}
    \begin{tabular}{r|r}
        $M\mapsto \ldots\hphantom{abcdfabcdf} $ & $\Qq_{A}\mapsto\ldots\hphantom{abcdfabcdf}$  \\
        \hline
       $M^{[a]}$ &  $\Qq_{A}^{[a]}$
       \\
       $\frac{f^+}{f^-} M$ & $f^{[A]_{\sh}}\Qq_{A}$
       \\
       $\TS(M)\equiv (M^{-1})^{\rm T}$ & $\TS(\Qq_{A})=\frac{(\Qq^A)^{[\rank]}}{\Qq_{\fullset}^{[\rank-|A|]}}$
    \end{tabular}\,.
\end{equation}

To get a better intuition about decomposition into $\gl_{\rank}$-invariants, we make a detour and consider the generating functional proposed in \cite{Kuniba:2001ub}. With some rearrangements, it can be written  as
\be
&&\bmB\frac 1{1-\sh^{-4}}\mB=\frac{\Pso^+}{\Pso^-}\left(\sum\limits_{a=0}^{\rank-2}  (-1)^a\sh^{-a} T_{a,1}\sh^{-a}\right)\frac{\Pso^-}{\Pso^+}+\ldots\,,
\nonumber
\\
&&\bmB=\prod_{a=1}^{\rank}(1-{\Lambda}_a\sh^{-2})\,,\quad
\mB=\prod_{a=1}^{\rank}(1-\sh^{-2}\bar\Lambda_a)\,,
\label{eq:genfun}
\ee
where all entries should be viewed as operators. Their right and left actions are defined by $\Lambda \triangleright f=f\triangleleft \Lambda=\Lambda f$, $\sh\triangleright f= f^+$, $f\triangleleft\sh= f^{-}$.

`Quantum eigenvalues’ $\Lambda_a$, $\bar\Lambda_a$, or rather their relevant combinations, can be found from functional relations
\begin{subequations}
\label{Baxters}
\begin{align}
\bmB \triangleright \Qq_a^+ &=0\,,&  \Qq_a^-\triangleleft \mB  &=0
\,,
\\
\left(\frac{\Qq^a}{\Qq_{\fullset}^+}\right)^{[-r+1]}\triangleleft \bmB  &=0\,,
&  
\mB\triangleright \left(\frac{\Qq^a}{\Qq_{\fullset}^-}\right)^{[r-1]} &=0\,.
\end{align}
\end{subequations}
There is a rich combinatorics that uses quantum eigenvalues and calculus of Young tableaux that allows computing various $\so_{2\rank}$ T-functions, see for instance \cite{Kuniba:1994na,Nakai2007}. We won't follow this approach but instead remark that \eqref{Baxters} have the form of $\gl_\rank$ Baxter equations which suggests that $\mB$ and $\bmB$ are generating functions for fundamental $\gl_\rank$ T-functions. By following this analogy, one finds
\be
\bmB = \sum_{a=0}^\rank(-1)^a\left(\Tg_{(1^a)}\right)^{[-r]}\sh^{-2a}\,,
\quad
\mB =\sum_{a=0}^\rank(-1)^a\sh^{-2a}\left({\bTg_{(1^a)}}\right)^{[r]}
\,.
\ee
We can note from \eqref{eq:auto} and \eqref{Baxters} that $\TS(\bmB)=\mB^{[2]}$ implying $\TS({\Tg_{(1^a)}}^{[a]})={\bTg_{(1^a)}}^{[2\rank+2-a]}$ which very explicitly reads $\bTg_{(1^a)}=\frac{{\Qq_\fs}^{[-2]}}{\Qq_\fs}{\Tg_{(1^{\rank-a})}}\!\!\!\!\!{\vphantom{{\Tg}}}^{[-2]}$. Using \eqref{eq:genfun}, we therefore see that $T_{a\leq r-2,1}$ are bilinear combinations of $\Tg_{(1^a)}$ with prefactors assembled from $\Phi$ (that in particular control the correct $\mathsf{U}(1)$ charge).

In general, $T_{a\leq r-2,s}$ are also given as bilinear combinations of $\Tg_{\lambda}$ with prefactors from $\Phi$, derivation of this fact is similar to \eqref{eq:V}: starting from \eqref{eq:Tvector}, we write $T_{a,s}$ as a sum over terms of type $\left(\mu_{AA'}^{[m']}\Qq^{A'B}\right)^{[r-1+s]}\left(\mu_{BB'}^{[m'']}\Qq^{B'A}\right)^{[-r+1-s]}$. Recursively using \eqref{eq:shiftmutrick}, we bring all $\mu_{ab}$ to the same value of shift and then the terms containing $\mu_{ab}$ will cancel each other. In the course of the recursion, it is not always that the second term in the \rhs of \eqref{eq:shiftmutrick} cancels and so one gets a multitude of terms $\Qq_{ab}$; the generated combinations of type ${\Qq^{a_1'\ldots b_1\ldots}}^{[r-1+s]}\left(\Qq_{a_1a_1'}^{[s_1]}\ldots\Qq_{b_1'b_1}^{[\tilde s_1]}\ldots\right){\Qq^{b_1'\ldots a_1\ldots}}^{[-r+1-s]}$ simplify to products of two $\Tg_\lambda$. This logic also provides an alternative approach to derive \eqref{eq:genfun}.

As an explicit example, consider $a=1$. One gets
\begin{align}
\label{eq:T1SQQ}
    T_{1,s} = 
        (-1)^{\rank-1}\,\sum_{k=0}^{s} {\qV^a}^{[-s-\rank+1]}\,\Qq_{ab}^{[2k-s]}\,{\qV^b}^{[s+\rank-1]}\,.
\end{align}
Then one uses $\Qq_{ab}=\Qq_a^+\Qq_b^--\Qq_a^-\Qq_b^+$ and for example ${V^{a}}^{[-s-\rank+1]}{\Qq_{a}}^{[2n-s+1]} = \frac{\Pso^{[1-s]}}{\Pso^{[-1-s]}}{\Tg_{(n)}}^{[-\rank+2-s]}$. Putting it together, 
\begin{equation}
\begin{aligned}
    T_{1,s} &=
    \frac{\Pso^{[1-s]}\Pso^{[-1+s]}}{\Pso^{[1+s]}\Pso^{[-1-s]}}\,\sum_{k=0}^{s}\left({\Tg_{(k)}}^{[-\rank+2-s]}\,{\bTg_{(s-k)}}^{[\rank-2+s]}
    -{\Tg_{(k-1)}}^{[-\rank+2-s]}\,{\bTg_{(s-k-1)}}^{[\rank-2+s]}\right)\,.  
\end{aligned}
\end{equation}
The $a=1$ case is also special because it~\mbox{allows} for another way to proceed. Using \eqref{eq:QPsiMap}, we can see that $(\qV^{a})^{[-\rank+1-s]}\Qq_{ab}^{[m]}(\qV^{b})^{[\rank-1+s]}=\frac{(-1)^{\rank-1}}{{V^{12\ldots\rank}}^{[m]}}|\qV^{b[\delta_{a,1}(s-m)-\delta_{a,\rank}(s+m)+m+\rank+1-2a]}|$
which~yields the Weyl-type expression of \cite{Ferrando:2020vzk}:
\begin{equation}\label{eq:WeylFormulaT1s}
    T_{1,s} = \sum_{k=0}^{s} \frac{|(\qV^{b})^{[2\delta_{a,1}(s-k)-2\delta_{a,\rank}k+2k-s+\rank+1-2a]}|}{(\qV^{12\ldots\rank})^{[2k-s]}}\,.
\end{equation}
It can be suggestively rewritten as
\begin{equation}
\label{eq:suggestive}
     T_{1,s} = \sum_{k=0}^{s}\frac{\Pso^{[2k-s+1]}\Pso^{[2k-s-1]}}{\Pso^+\Pso^-}T^V_{(s/2,k-s/2,\dots,k-s/2,-s/2)}
\end{equation}
for $T^V_{\lambda} = \Pso^+\Pso^- |{(V^{a})^{[2\hat{\lambda}_b]}}|$. T-functions $T^V_{\lambda}$ are also $\gl_{\rank}$-invariants; they can be obtained from $\Tg_{\lambda}$ through the chain of automorphisms \eqref{eq:auto}: $T^V_{\lambda}=\varphi(\Tg_{\lambda})$ where $\varphi:M\mapsto \left(\frac{\Pso^2}{\Pso^{[2]}\Pso^{[-2]}}\right)^{+} \TS(M^{[-\rank]})$.

The spinor T-functions can be also decomposed starting from  \eqref{eq:Tspinor},\eqref{eq:Tspinor2} using similar techniques. There is a slight rank dependence in the intermediate computations, we start from the even $r$ case. The key observation is to rewrite \eqref{eq:GammaToForms} in the form
\begin{align}
\label{eq:SpinorTCalculation}
    \QDirac^{[-m]} C\Projp \QDirac^{[m]}
    = \frac{1}{\frac{r}{2}!}\frac{\left(\mu^{[m]}-\mu^{[-m]}\right)^{\frac{\rank}{2}}}{\Pso^{[m]}\Pso^{[-m]}}
    =\frac{1}{\frac{\rank}{2}!}\frac{\left([m]_\sh {\Qq_{(2)}}\right)^{\frac{r}{2}}}{\Pso^{[m]}\Pso^{[-m]}}\,,
\end{align}
where $\mu=\sum\limits_{a<b}\mu_{ab}\theta^a\theta^b$ and $\Qq_{(k)}=\sum\limits_{|A|=k} \Qq_{A}\theta^A$.
Then
\begin{align}
\label{eq:SpinorTEven}
    {T_{r-1,s}} &=\frac 1{\frac{r-2}{2}!}{\Qq^{[\rank-1+s]}_{(1)}}{} \frac{\left([\rank\!-\!3+s]_{\sh}\Qq_{(2)}\right)^{\frac{r-2}{2}}}{\Pso^{[r-1+s]}\Pso^{[-r+1-s]}}\Qq^{[-\rank+1-s]}_{(1)}\,,
    &
    {T_{\rank,s} } &= \frac{1}{\frac{r}{2}!}\frac{\left([\rank\!-\!1+s]_{\sh}\Qq_{(2)}\right)^{\frac{r}{2}}}{\Pso^{[r-1+s]}\Pso^{[-r+1-s]}}\,.
\\
\shortintertext{
For odd ranks, the equivalent computation yields}
\label{eq:SpinorTOdd}
    {T_{\rank-1,s}} &= \frac {1}{{\frac{r-1}{2}!}}{\Qq^{[\rank-1+s]}_{(1)}} \frac{\left([r-1+s]_{\sh}\Qq_{(2)}\right)^{\frac{r-1}{2}}}{{\Pso^{[\rank-1+s]}\Pso^{[-\rank+1-s]}}}\,,&
    {T_{\rank,s} } &=
    \frac{1}{\frac{\rank-1}{2}!}\frac{\left([\rank-1+s]_{\sh}\Qq_{(2)}\right)^{\frac{\rank-1}{2}}}{\Pso^{[\rank-1+s]}\Pso^{[-\rank+1-s]}} {\Qq^{[-\rank+1-s]}_{(1)}}\,.
\end{align}
When $s=1$, \eqref{eq:SpinorTEven} and \eqref{eq:SpinorTOdd} agree with the expressions in \cite{Ferrando:2020vzk}. The results \eqref{eq:SpinorTEven} and \eqref{eq:SpinorTOdd} can be assembled as
\begin{equation}
\label{eq:Tassp}
    T_{a,s} = \frac{\Pso^{[\rank-1]}\Pso^{[-\rank+1]}}{\Pso^{[\rank-1+s]}\Pso^{[-\rank+1-s]}}\sum_{\lambda}\bTg_{\lambda}\,, \quad
    a=\rank-1,\rank\,,
\end{equation}
where the sum runs over all $\gl_\rank$-irreps $\lambda$ that occur in the decomposition of the corresponding $\so_{2\rank}$-irrep.

Existence of linear decompositions \eqref{eq:suggestive}, \eqref{eq:Tassp} can be explained by the fact that $\so_{2\rank}$-irrep with the highest weight $s\,\omega_a$ for $a=1,\rank-1,\rank$ can be lifted to the evaluation irreps of a quantum algebra~\footnote{This property is the same for Yangian and quantum affine cases.}, and so the decomposition of such representations follows the character decomposition on the Lie algebra level. On the contrary, if $2\leq a\leq r-2$, the $\so_{2\rank}$-irreps are not quantisable---the corresponding KR-modules are larger as vector spaces---and we do not expect similar simplifications, we also checked numerically that this indeed does not happen.

The advantage of $\Tg_\lambda$-decomposition is that the result depends only on $\Qq_a$ and their Wronskian determinants (all the prefactors containing $\Pso$ can be rewritten as rational combinations of $\Qq_{\fullset}$). However, we saw that these formulae are telescoped-up versions of inner products \eqref{eq:Ts} that are much more compact expressions. In addition to $\Qq_a$ and $\Pso$, \eqref{eq:Ts} depend also on $\mu_{ab}$, but this is a little price to pay because one needs to find $\mu_{ab}$ anyway when solving the Q-system equations. Hence we believe that, although $\Tg_\lambda$-decompositions have a conceptual value in illuminating various aspects of $q$-character theory, \eqref{eq:Ts} is a more economical way to compute $T_{a,s}$ in practice.

\section{Rational Spin Chain}\label{sec:RationalSpinChains}
To apply the developed formalism for a concrete physical model one should require specific analytic features of Q-functions. This requirement called {\it analytic Bethe Ansatz} by analogy with \cite{Reshetikhin:1983vw} strongly constrains possibility to satisfy the functional relations of the extended Q-system. Each possible solution corresponds then to a physical state of the model, potentially modulo symmetries on both sides of the correspondence. As a concrete example, we shall do rational spin chains.

\subsection{Analytic Bethe Ansatz}\label{sec:5a}
The explicit demand on Q-functions to explore is \cite{Ekhammar:2020enr}
\be\label{eq:AnalyticalAnsatz}
    Q_{(a)}^{\mathfrak{i}} \propto  \dressing_{a}\times q_{(a)}^{\mathfrak{i}}\,.
\ee
Here $\propto$ means `equal up to a $u$-independent normalisation', and the normalisation is chosen by requiring the compatibility with \eqref{eq:DynkinDiagramQQ}; $q_{(a)}^{\mathfrak{i}}$ are monic polynomials in $u$, it is natural to call them Baxter polynomials with their zeros being Bethe roots, and $\dressing_{a}$ are dressing factors satisfying
\be
\prod_{b=1}^\rank\sigma_b^{-[C_{ab}]_\sh}=P_a\,,
\ee
where $C_{ab}$ is the $\so_{2r}$ Cartan matrix and $P_a$ are some fixed monic polynomials in $u$. The structure of dressing factors was tailored for QQ-relations \eqref{eq:DynkinDiagramQQ} to become
\be\label{eq:DynkinDiagramqq}
W(q_{(a)}^{\weyl(1)},q_{(a)}^{\weyl(2)})\propto P_a  \prod_{b\sim a} q_{(b)}^{\weyl(1)}\,.
\ee
One can choose $\weyl={\rm Id}$ and apply the standard argument to get nested Bethe equations
\be
\label{eq:nBAE}
\frac{q_{(a)}^{1[+2]}}{q_{(a)}^{1[-2]}}\prod_{b\sim a}\frac{q_{(b)}^{1[-1]}}{q_{(b)}^{1[+1]}}=-\frac{P_a^+}{P_a^-}\quad\quad\text{at zeros of }q_{(a)}^1\,.
\ee
Comparing to nested Bethe equations from the literature \cite{Reshetikhin:1983vw,OGIEVETSKY1986360},
we see that the chosen analytic Bethe Ansatz is supposed to describe the spectrum of rational spin chains in finite-dimensional representations of Yangian $\Yangian(\so_{2\rank})$, and with periodic boundary conditions used in the construction of the Bethe algebra.  In particular, it should describe a spin chain of length $L$ with sites in the vector representation of $\so_{2\rank}$ if 
\begin{equation}
\label{eq:Dvec}
    P_1 = \prod_{l=1}^{L}(u-\theta_l)\,,
    \quad
    P_{a\neq 1} = 1\,.
\end{equation}
Numbers $\theta_l$ are known as inhomogeneities. The homogeneous case $\theta_l=0$ corresponds to a chain with nearest-neighbour interactions.

Every finite-dimensional representation of $\Yangian(\so_{2\rank})$ has the highest weight state dubbed the ferromagnetic vacuum which is uniquely characterised by the collection of Drinfeld polynomials. From \eqref{eq:nBAE}, one identifies that $P_a$, $a=1,\ldots,\rank$, are Drinfeld polynomials \cite{Kirillov1990RepresentationsOY}.

\medskip
\noindent
For generic values of parameters, \eqref{eq:DynkinDiagramqq} for $\weyl={\rm Id}$ follows from the nested Bethe equations and then using arguments of the reproduction procedure in \cite{Mukhin:2005aa} we conclude that all Q-functions on the Weyl orbit satisfy the analytic Bethe Ansatz; finally, with help of \eg \eqref{eq:TensorQ} we show that all Q-functions of the extended Q-system  satisfy the analytic Bethe Ansatz. Hence the analytic Bethe Ansatz is equivalent to the nested Bethe equations, however in generic position.

We included the discussion of nested Bethe equations because it is a useful connection point with other works. However, we have no goal to rely on them, instead we shall use Wronskian Bethe equations.

\medskip
\paragraph*{Physical vs Langlands-dual $\so_{2\rank}$.}
\label{par:PLd}
Before going into the computational hurdles, we need to further discuss the setting of the problem. There are {\it two different} $\so_{2\rank}$ algebras involved. One $\so_{2\rank}$ algebra is the physical symmetry. It is a subalgebra of $\Yangian(\so_{2\rank})$ which commutes with the Bethe algebra and in particular with the spin chain Hamiltonian. As a module under this $\so_{2\rank}$ action, the Hilbert space decomposes as
\be\label{eq:HilberSpaceDecomposition}
    \mathcal{H} =  \bigoplus\limits_{\lambda} \multiplicity_{\lambda} \, L(\lambda)\,,
\ee
where $L(\lambda)$ are $\so_{2\rank}$ irreps with the highest weight $\lambda$ and $\multiplicity_{\lambda} $ is the multiplicity of their occurrence. The multiplicity of the ferromagnetic vacuum multiplet is $\multiplicity_{\lambda_{\rm max}}=1$; the weight of the ferromagnetic vacuum is encoded in the degrees of Drinfeld polynomials:
\be
\label{eq:HilberSpaceDecompositionaaa}
\lambda_{\rm max}=\sum_{a=1}^{\rank}\omega_a\,{\rm deg}\,P_a\,.
\ee

The other $\so_{2\rank}$ algebra is the one which we worked with in all the previous sections. We remind that it acts on $Q_{(a)}$ transforming them as vectors in the $a$-th fundamental representation $L(\omega_a)$ and this action is the symmetry of the equations of the extended Q-system. In the decomposition $Q_{(a)}=\sum_{\mathfrak{i}}Q_{(a)}^{\mathfrak{i}}\basvec_{\mathfrak{i}}^{(a)}$, vectors $\basvec_{\mathfrak{i}}^{(a)}$ are weight basis vectors of $L(\omega_a)$ of weight $\gamma_{(a)}^{\mathfrak{i}}$, and it is the convention that $\gamma_{(a)}^{1}\equiv \omega_a$ is the highest weight. The following rules for computing $\gamma_{(a)}^{\mathfrak{i}}$ apply
\begin{align}
    \qV^{\pm a} \to \gamma_{(1)}^{\pm a}= \pm \obasis_{a}\,,\quad
    \QPS_{a} \to \gamma_{(\rank-1)}^{\rank-a+1} = \omega_\rank - \obasis_{a}\,,\quad
    \QPS_{\emptyset} \to \gamma_{(\rank)}^{1} = \omega_\rank\,.
\end{align}
More generally, the weight of $\basvec_{\mathfrak{i}}^{(a)}$ associated to  $\QPS_{A}$ is $\frac 12\big(\sum\limits_{a\notin A}\obasis_a -\sum\limits_{a\in A}\obasis_{a}\big)$, and to $V^I$ is $\sum\limits_{a\in I}\obasis_a-\sum\limits_{-a\in I}\obasis_a$. Here we used the standard orthonormal basis $\obasis_a$, $(\obasis_a,\obasis_b) = \delta_{ab}$, of the weight space. The fundamental weights are expanded as $\omega_{a} = \sum\limits_{b=1}^{a}\obasis_{b}$ if $a<\rank-1$, $\omega_{\rank}=\frac{1}{2}\sum\limits_{a=1}^{\rank}\obasis_a$ and $\omega_{\rank-1}=\omega_\rank-\obasis_\rank$ in this basis.

\medskip
\noindent
Each physical multiplet $L(\lambda)$ is expected to be represented by a solution of the extended Q-system and hence we expect $\multiplicity_{\lambda}$ solutions representing multiplets of weight $\lambda$. There is one solution describing the ferromagnetic vacuum multiplet, it has the simplest possible highest-weight components of Q-functions: $q_{(a)}^1=1$. Other multiplets---excitations around the vacuum---have non-trivial polynomials $q_{(a)}^1$ whose degrees are known as magnon numbers. Overall, degrees of Baxter polynomials for a multiplet $L(\lambda)$ are decided as follows
\begin{align}
\label{eq:degrees}
{\rm deg}\,q_{(a)}^{\mathfrak{i}}=(\omega_{a},\lambda_{\rm max}+\rho)-(\gamma_{(a)}^\mathfrak{i},\lambda+\rho)\,,
\end{align}
where $\rho=\sum\limits_{a}\omega_a=\sum\limits_{a=1}^{\rank-1} (\rank-a)\obasis_a$ is the Weyl vector.

In the `scalar products' $(\omega_{a},\lambda_{\rm max}+\rho)$ and $(\gamma_{(a)}^{\mathfrak{i}},\lambda+\rho)$, vectors $\lambda_{\rm max},\lambda,\rho$ are the ones from the weight space of $\so_{2\rank}$ which is the physical symmetry, while $\omega_a,\gamma_{(a),\mathfrak{i}}$ are from the weight space of $\so_{2\rank}$ which is the Q-system symmetry. Taking scalar products of vectors from different spaces is of course nonsense and we use this notation only because it produces correct numbers within the chosen computational prescriptions. A conceptually accurate interpretation of \eqref{eq:degrees} is that the `scalar product' is in fact the pairing between dual vector spaces, and the two $\so_{2\rank}$ algebras are Langlands duals of one another.

\subsection{Implementation via Wronskian Bethe equations}
\label{sec:VB}
We shall search for solutions of the extended Q-system by first focusing on the triple of spinor Q-functions. For them, the analytic Bethe Ansatz reads $\Psi_\emptyset\propto\sigma_{\rank}\times \qPS_0$, $\Psi_{a}\propto \sigma_{\rank-1}\times \qPS_a$, $\Psi_{ab}\propto \sigma_{\rank}\times\qPS_{ab}$. Let $\degQ_a$ be the degree of $\psi_a$ and $\degQ_0$ the degree of $\psi_0$. Parameterise these polynomials as
\begin{align}\label{eq:Parameterisation}
    \qPS_{a} = u^{\degQ_{a}}+ \sum^{\degQ_{a}-1}_{l=0}c_{a,l}\, u^{l} \,,
    \quad
    \qPS_{0} = u^{\degQ_{0}}+ \sum^{\degQ_{0}-1}_{l=0}c_{0,l}\,u^{l}\,.
\end{align}
Our goal shall be to derive equations on $c_{a,l}$ for $a=0,\ldots,\rank$.

The Q-system is covariant under action of $\so_{2r}$. Generally, one can introduce a $u$-dependent holomorphic connection making this covariance into a gauge symmetry. As was mentioned in the introduction, we work in a gauge where the connection is trivial or more accurately where the parallel transport from $u$ to $u+\ii$ is trivial. Assuming this, we need to fix the residual gauge freedom. Ansatz \eqref{eq:AnalyticalAnsatz} allows for only global $u$-independent $\so_{2r}$ transformations, and only transformations from the Borel subalgebra remain once \eqref{eq:degrees} is imposed. Symmetry transformations from Cartan subalgebra can rescale \eg $\Psi_a$ by arbitrary constants, but this is hidden in the $\propto$ sign of \eqref{eq:AnalyticalAnsatz}. We are left with the nilpotent subalgebra only, and we partially fix this freedom by setting
\begin{subequations}
\begin{align}
\label{eq:c0}
    &c_{a,\degQ_b} = 0\,,
    \quad
    a<b\leq \rank\,.
\end{align}
\end{subequations}

It is finally the time to consider Wronskian Bethe equations \eqref{eq:SpinorQSystem}. Start by rewriting \eqref{eq:U1Conditions} in terms of polynomial functions
\begin{equation}\label{eq:FiniteDifferenceWronsk}
    \begin{split}
    W(\qPS_{ab},\qPS_{0}) \propto \SourcePS_{ab}\,,
    \quad
   \SourcePS_{ab} = \frac{P_r}{P_{r-1}}W(\qPS_{a},\qPS_{b})\,,    
\end{split}
\end{equation}
and then formally invert this relation to find $\qPS_{ab}$:
\begin{equation}\label{eq:InvertedSpinorWronskian} 
    \qPS_{ab} \propto \qPS_0\frac{1}{D-D^{-1}} \frac{\SourcePS_{ab}}{\qPS_0^{+}\qPS_{0}^-}\,,
\end{equation}
where the overall normalisations are fixed from the requirement that $\psi_{ab}$ are monic.
Expand \rhs around $u=\infty$, to this end one can use $D= e^{\frac{\ii}{2}\partial_u}$ and so
\begin{equation}
    \frac{1}{D-D^{-1}} = -\frac{\ii}{\partial_{u}}-\frac{\ii}{24}\partial_u - \frac{7 \ii}{5760}\partial^3_u + \dots\,.
\end{equation}
Since $\qPS_{ab}$ has to be polynomial, the obtained Laurent series in \eqref{eq:InvertedSpinorWronskian} must terminate and we can discard terms of order $\mathcal{O}(\frac{1}{u})$, for the same reason we can set $\frac 1{\partial_{u}}\frac 1u=0$, for other powers $\frac 1{\partial_{u}}$ is the usual integration $\frac 1{\partial_{u}}u^n=\frac 1{n+1}u^{n+1}$. The integration constant can be chosen arbitrarily because it is yet unused piece of the global $\so_{2\rank}$ symmetry: we can subtract any multiple of $\qPS_0$ from $\qPS_{ab}$. We fix this symmetry by setting the coefficient in $\qPS_{ab}$ in front of $u^{\degQ_{0}}$ to zero. All symmetries are fixed after this step.

Now we plug the obtained $\qPS_{ab}$ back into \eqref{eq:FiniteDifferenceWronsk}, this gives non-trivial equations for the remaining $c_{a,l}$. We combine them with~\eqref{eq:DeterminantCondition}
\begin{equation}
\label{eq:W}
    W(\qDirac_{1},\ldots,\qDirac_{\rank}) \propto \qDirac_{0}^{[\rank-2]_{\sh}}\prod_{a=1}^{\rank-1} P^{[a]_{\sh}}_{a}\,.
\end{equation}

Equations \eqref{eq:W} and \eqref{eq:FiniteDifferenceWronsk} is an explicit implementation of the analytic Bethe Ansatz applied to Wronskian Bethe equations \eqref{eq:SpinorQSystem}. Let us assume we have solved them, \ie found $c_{a,l}$ and hence $\Psi_{\es},\Psi_{a},\Psi_{ab}$. Does it mean that we solved the extended Q-system as well? Indeed,  it is possible to compute other Q-functions from $\Psi_{\es},\Psi_{a},\Psi_{ab}$ and, by Theorem~\ref{thm:derivingextendedQsystem}, they will satisfy the extended Q-system relations. But this argument is not sufficient. We need also to verify whether all the Q-functions computed in this way conform with analytic Bethe Ansatz \eqref{eq:AnalyticalAnsatz}, otherwise the solution of Wronskian Bethe equations should be deemed nonphysical.

One can readily notice that \eqref{eq:W} and \eqref{eq:FiniteDifferenceWronsk} cannot reflect all details of a sophisticated enough spin chain because Drinfeld polynomials $P_a$, which could be in principle $\rank$ independent functions, appear only through combinations $P_{\rank}/P_{\rank-1}$ and $\prod\limits_{a=1}^{\rank-1} P^{[a]_{\sh}}_{a}$. We mark off important cases  when these combinations nonetheless provide the sufficient information---when $P_r=1$ for $r\neq 1,a-1,a$ and $P_{\rank-1}$ and $P_{\rank}$ are co-prime. Such cases shall be referred to as basic spin chains. They include for instance chains where every node is in the vector representation, or in a Weyl spinor representation of the same chirality. 

For non-basic spin chains, we need further detailing. To this end we notice that tensor Q-functions $V^A$ computed via \eqref{eq:VectorUpperFromSpinor} do not automatically satisfy the analytic Bethe Ansatz in consequence of Wronskian Bethe equations. Indeed, with the Ansatz $V^{A}\propto \sigma_{|A|}v^A$, we get 
$
v^A\propto \frac{\epsilon^{Aa'_1\ldots a'_k}W(\qDirac_{a'_1},\ldots,\qDirac_{a'_k})} {\qDirac_0^{[k-2]_{\sh}}\prod\limits_{b=r-k+1}^{r-1} P_{b}^{[b+k-r]_{\sh}} }\,.
$
For $v^A$ to be polynomials for $1\leq |A|\leq r-2$, the following requirement should be
additionally imposed:
\begin{align}\label{eq:SpinorFromVec}
    {\rm Remainder}
\left(\frac{W(\qDirac_{a_1},\ldots,\qDirac_{a_k})} {\qDirac_0^{[k-2]_{\sh}}\prod\limits_{b=r-k+1}^{r-1} P_{b}^{[b+k-r]_{\sh}} }\right)=0\,\quad\quad\text{for $2\leq k\leq r-1\,.$}
\end{align}
Practical interpretation is the following: if we impose only Wronskian Bethe equations \eqref{eq:W} and \eqref{eq:FiniteDifferenceWronsk}, this would correspond to considering an effective basic spin chain with $P_{\rank}^{\rm eff}/P_{\rank-1}^{\rm eff}=P_{\rank}/P_{\rank-1}$ and $P_1^{\rm eff}=(P_{\rank-1}^{\rm eff})^{-[r-1]_{\sh}}\prod\limits_{a=1}^{\rank-1} P^{[a]_{\sh}}_{a}$. We shall get a finite number of solutions but this number is excessive compared to the dimension of Hilbert space if the true spin chain is not basic. For  description of the latter one needs to filter the obtained solutions and keep only those that satisfy \eqref{eq:SpinorFromVec}; we refer to \eqref{eq:SpinorFromVec} as `kinematic constraints' for this reason. They are written as an overdetermined system for all possible choices of $a_1,\ldots,a_k$ but using them in this way turns out to be time-efficient in explicit computations.

\medskip
\noindent
We are now ready to discuss analytic properties of all Q-functions in consequence of \eqref{eq:FiniteDifferenceWronsk}, \eqref{eq:W}, and \eqref{eq:SpinorFromVec}.
\begin{lemma}
\label{lemma3}
In consequence of \eqref{eq:FiniteDifferenceWronsk} only, and in generic position, $W(\qDirac_{a_1},\ldots,\qDirac_{a_k})$ is divisible by $\qDirac_0^{[k-2]_{\sh}}$~for~$k\leq r$.
\end{lemma}
\begin{proof}
For $k=3$: $W(\psi_{a},\psi_{b},\psi_{c})=\frac{W(W(\psi_{a},\psi_{b}),W(\psi_{a},\psi_{c}))}{\psi_{a}}=\left(\frac{P_{r-1}}{P_{r}}\right)^{[2]_{\sh}}\frac 1{\psi_a}W(\psi_{0},\psi_{ab},\psi_{ac})\psi_{0}$ from where divisibility by $\psi_0$ follows assuming $\psi_0$ and $P_r^{[2]_\sh}\psi_a$ do not have common zeros. For  $k>3$, we use $W(\psi_{a_1},\ldots,\psi_{a_k})=\frac{W(W(\psi_{a_1},\ldots,\psi_{a_{k-2}},\psi_{a_{k-1}}),W(\psi_{a_1},\ldots,\psi_{a_{k-2}},\psi_{a_{k}}))}{W(\psi_{a_1},\ldots,\psi_{a_{k-2}})}$ and confirm divisiblity by $\psi_0^{[k-2]_\sh}$ using induction in $k$.
\end{proof}
This was a technical observation demonstrating that $\qDirac_0^{[k-2]_{\sh}}$ can be, under the assumptions of the lemma, pulled out from the remainder in \eqref{eq:SpinorFromVec}  and hence \eqref{eq:SpinorFromVec} in the case of basic spin chains offers no additional information compared to Wronskian Bethe equations.

\begin{lemma}[{Sufficient condition for polynomiality of $q_{(a)}^{\mathfrak{i}}$}]
\label{lemma:sufficient}
Assuming Wronskian Bethe equations \eqref{eq:FiniteDifferenceWronsk}, \eqref{eq:W} together with kinematic constraints \eqref{eq:SpinorFromVec}, all the Q-functions of the extended Q-system computed from $\Psi_\es,\Psi_a,\Psi_{ab}$ satisfy analytic Bethe Ansatz \eqref{eq:AnalyticalAnsatz} if $\psi_0$ does not contain zeros separated by $\ii$.
\end{lemma}
The proof is in Appendix~\ref{app:techproofs}. 

It is natural to assume that zeros of $\psi_0$ non-trivially depend on the input data---zeros of Drinfeld polynomials a.k.a. inhomogeneities---and thus we can get zeros of $\psi_0$ not differ by $\ii$ by changing inhomogeneities a bit. Likewise it is assumed that genericity assumptions used in the proof of Lemma~\ref{lemma3} can be achieved by modifying inhomogeneities as well. Although we have a solid experimental evidence for this assumption we cannot offer a systematic rigorous proof.

With this assumption on inhomogeneities, we conclude that analytic Bethe Ansatz follows from \eqref{eq:FiniteDifferenceWronsk}, \eqref{eq:W}, \eqref{eq:SpinorFromVec} for generic values of parameters and in particular, the proposed equations can be considered as an alternative equivalent to nested Bethe equations. The observations we made are fully analogous to the ones between Wronskian Bethe equations and nested Bethe equations for $\mathfrak{gl}_{\rank}$ chains \cite{MTV,Chernyak:2020lgw}. Moreover, the analogy includes the role of kinematic constraints: The focus of \cite{MTV,Chernyak:2020lgw} was on a spin chain with nodes in vector representations which is an example of a basic spin chain. If one wants to explore more complicated examples, a supplement with kinematic constraints similar to \eqref{eq:SpinorFromVec} is needed for $\mathfrak{gl}_{\rank}$ case as well, such constraints were for instance a part of the AdS/CFT spectrum computation \cite{Marboe:2017dmb}.

\setlength{\textfloatsep}{0.8\baselineskip plus 0.2\baselineskip minus 0.2\baselineskip}
\subsection{Explicit results}\label{sec:SolvingQSystem}
Here is the distilled summary of what is being done for the explicit computations: Use $P_a$---Drinfeld polynomials---and $\lambda$---the weight of physical $\so_{2\rank}$-irrep in \eqref{eq:HilberSpaceDecomposition}---as an input, derive polynomial equations for $c_{a,l}$ following the routine of subsection~\ref{sec:VB} from \eqref{eq:Parameterisation} to \eqref{eq:W}, optionally add \eqref{eq:SpinorFromVec}, and ask a computer software to solve the obtained system. The corresponding code in {\it Mathematica} is provided in the ancillary notebook.
Among numerous tests, we made a systematic study of spin chains described by Drinfeld polynomials $P_b=u^{L\,\delta_{ba}}$, \ie of homogeneous chains with all $L$ nodes in the same $a$-th fundamental representation of Yangian, $a=1,\ldots,\rank$. The summary is given in tables~\ref{tab:D4Solutions} and~\ref{tab:D5Solutions} where we provide the number of solutions and the time required for generation of them analytically, in terms of algebraic numbers. 
\newsavebox{\tmpbox}
\sbox{\tmpbox}{%
\begin{minipage}[b]{0.48\linewidth}
    \begin{tabular}{c|c|c|c|c}
         $L\backslash a$ & 1 & 2 & 3 & 4 \\
         \hline
        2 & 3(0.9s)  & 10(11.5s) & 3(2.0s) & 3(2.4s) \\
        3 & 7(2.0s) & 68(95.0s) & 7(5.0s) & 7(11.0s) \\
        4 & 26(11.2s)& 631(1177s)& 26(29.6s)& 26(120s)\\
        5 & 85(28.0s)& - & 85(78s)& 85(322s)\\
        6 & 365(79s)& - & 365(1435s)& 365(2278s)  \\
        7 & 1456(1483s)& - & - & - 
    \end{tabular}
    \vfill
    {\captionof{table}{\label{tab:D4Solutions}$\so_{8}$}}
    \end{minipage}
}
\begin{figure}[h]
\centering
\usebox{\tmpbox}
    ~
    \begin{minipage}[b][{\ht\tmpbox}]{0.48\linewidth}
    \begin{tabular}{c|c|c|c|c|c}
        $L\backslash a$ & 1 & 2 & 3 & 4 & 5\\
        \hline
        2 & 3(2.4s) & 9(21.1s) & 20(108s) & 3(8s) & 3(18.5s)  \\
        3 & 7(4.7s) & 60(176s) & - & 9(22.8s) & 9(136s)\\
        4 & 25(21.0s) & - & - & 42(325s) & 42(1571s) \\
        5 & 82(215s) & - & - & - & - \\
    \end{tabular}
    \vfill
    \captionof{table}{ \label{tab:D5Solutions}$\so_{10}$}
    \end{minipage}
\end{figure}

The resulting procedure turned out to be remarkably efficient, and there should be ample room to further improve the performance. In particular, it would be interesting to develop ideas similar to those of \cite{Marboe:2016yyn}.

The proposed approach always gave the expected number of solutions $\multiplicity_{\lambda}$, and this should be contrasted with other approaches that have known shortcomings: Probably the most known other approach---based on nested Bethe equations---is not supplied with a prescription how to treat cases with coinciding Bethe roots and how to treat $0/0$-type cases (exceptional solutions) which complicates counting; The other approach is to solve QQ-system on Dynkin diagram that is relations \eqref{eq:DynkinDiagramqq} with $\weyl={\rm Id}$ involving only $q_{(a)}^{1}$ and $q_{(a)}^{2}$. This one is typically over-complete. Indeed, the exceptional solutions where both terms of the Wronskian determinants vanish simultaneously at some point can be erroneously accepted while being non-physical. Our point of view is that demanding analytic Bethe Ansatz for $Q_{(a)}^{\mathfrak{i}}$ beyond $\mathfrak{i}=1,2$ is what selects physical solutions resolving over-completeness of QQ-system on Dynkin diagram.

The simplest instance where QQ-system on Dynkin diagram is not enough is $L=2$ homogeneous vector $\mathfrak{so}_8$ spin chain. It has one copy of symmetric traceless, anti-symmetric, and trace representations: $
    \mathcal{H}= L(\omega_2)\oplus L(\omega_1) \oplus L(0)\,.
$ Solving the QQ-system on Dynkin diagram for  $\lambda=0$, we find four solutions
\begin{subequations}
\begin{align}
    &\qv^1 \propto \{u^2-\frac{1}{8},u^2,u^2+\frac{1}{4},u^2+\frac{1}{4}\}\,, \\
    &\qv^{12} \propto \{u^2-\frac{1}{8},u^2+\frac{1}{4},u^2,u^2+\frac{1}{4}\}\,,\\
    &\psi_{4} \propto \psi_{0}\propto \{u,u,u,u\}\,,
\end{align}
\end{subequations}
however there should be only one physical solution. By solving Wronskian Bethe equations, we indeed observe that only one solution is valid, with $v^1=u^2-\frac 18$.

In the example above, all the problematic for nested Bethe equations solutions (\ie those with double zeros or zeros at $u=\pm \ii/2$) were discarded, but of course they do appear as physical solutions in other examples. There are plenty of them in $\sl_{\mathsf{n}}$ sectors as we know from experience with $\gl_{\mathsf{m}|\mathsf{n}}$ spin chains, but they exist in full $\so_{\mathsf{2\rank}}$ sectors as well. One of the first examples is the $\so_8$ homogeneous vector spin chain of length four. For $\lambda=0$, there is an exceptional solution featuring $
    \qv^1 \propto (u^2+\frac{1}{4})(u^2-\frac{5}{28})
    $. In nested Bethe equations, deciding whether this solution is physical requires going through a regularisation hurdle whereas our approach offers it as the valid solution automatically. 
\setlength{\textfloatsep}{0.8\baselineskip plus 0.2\baselineskip minus 0.2\baselineskip}

\subsection{Completeness and faithfulness conjecture}
\label{sec:CF}
\label{sec:cfconjecture}
In the conducted experiments, not only generic position values of inhomogeneities were tested meaning we considered situations where physically valid solutions contained Q-functions with coinciding Bethe roots or roots separated by $\ii$. In fact, a homogeneous chain was often not a generic position  either.  Although we could not have probably tested all possibilities, the achieved success in all trials is a good indication that the `generic position' clause can be waived in Wronskian Bethe equation describing the spectrum statements and hence we conjecture that the developed formalism can be used to formulate the rigorous completeness theorems. This is also supported by previous experience with $\mathfrak{gl}_{\rank}$ chains \cite{MTV,Chernyak:2020lgw}, the results of these papers invited us to investigate the completeness features at the first place.

The completeness question consists of two parts. The one, which \cite{Chernyak:2020lgw} refers to as faithfulness, is to prove that the Q-system is indeed isomorphic to the Bethe algebra generated by transfer matrices. More accurately, the isomorphism to establish is between Wronskian Bethe algebra---the polynomial ring $\mathcal{W}=\mathbb{C}[c_{a,l}]/\langle \text{Wronskian Bethe equations}\rangle$ computed at weight $\lambda$---and the Bethe algebra restricted to the highest-weight subspace of weight $\lambda$. An important argument for this assesment that would prove surjectivity of the map from $\mathcal{W}$ to the Bethe algebra and facilitate proving injectivity is to construct Q-functions as eigenvalues of Baxter Q-operators. For $\so_{2\rank}$ chains, this construction and in particular the derivation of the analytic Bethe Ansatz conditions is not yet fully accomplished, but an important progress was made in \cite{Frassek:2020nki}. Combining these results with combinatorics of Theorem~\ref{thm:derivingextendedQsystem} seems to be a sufficient toolkit to finalise the Q-operator question, at least for basic spin chains, but it is yet to be done.

We did not do experimental checks on the level of Q-operators. However, where it was technically possible, we computed the eigenvalues of transfer matrices as operators acting on the spin chain constructed explicitly from Lax operators and confirmed them to match with T-functions computed via \eqref{eq:Ts}.

\medskip
The second part is the actual completeness statement. Within the established isomorphism according to the faithfulness property, it can be equivalently formulated either as the fact that the Bethe algebra contains all possible simultaneously commuting operators, \ie it is a maximal commutative algebra, or that the algebraic number of solutions, \ie the dimension of the Wronskian Bethe algebra $\mathcal{W}$ as a vector space over $\mathbb{C}$, is $\multiplicity_{\lambda}$. Our main argument in favour of completeness is reproducing numbers  $\multiplicity_{\lambda}$ in various experiments and we notice that even generic position points are worthwhile of checking as, to our knowledge, completeness of nested Bethe equations was not proven for the $\so_{2\rank}$ case even under generic position assumptions.

We explain now more accurately how the Hilbert space $\mathcal{H}$ of spin chain is constructed, from here it follows what values of $\multiplicity_{\lambda}$ we are checking to reproduce on experiments. The $a$-th fundamental representation of Yangian is, by definition, the irrep with $P_b=(u-\theta)^{\delta_{ba}}$, $b=1,\ldots,\rank$. We denote it by $\qrep_a(\theta)$.
In contrast to the $\sl_{\rank}$ case, these fundamental representations can be {\it reducible} as representations of $\so_{2r}$. The description of $\qrep_{a}(\theta)$ as an $\so_{2r}$-module reads
\begin{subequations}\label{eq:YangianDecomposition}
\begin{align}
\label{eq:74a}
    &\qrep_{a}(\theta) = L(\omega_a)\oplus L(\omega_{a-2})\oplus \dots \,,
    \quad
    a = 1,2,\dots,\rank-2\,, \\
    &\qrep_{{\rank-1}}(\theta) = L(\omega_{\rank-1})\,,
    \quad
    \qrep_{\rank}(\theta) = L(\omega_\rank)\,,
\end{align}
\end{subequations}
where the sum in \eqref{eq:74a} terminates either by $L(\omega_1)$ for odd $a$ or by the trivial representation $L(0)$ for even $a$.

By taking tensor products of the fundamental representations, one builds a spin chain of length $L$, with Hilbert space
\be
\label{eq:H}
\mathcal{H}=\bigotimes\limits_{\ell=1}^L \qrep_{a_\ell}(\theta_{\ell})\,.
\ee
Its highest-weight vector is characterised by $P_a=\prod\limits_{\ell=1}^{L}(u-\theta_{\ell})^{\delta_{a a_{\ell}}}$. Typically $\mathcal{H}$ is an irreducible Yangian representation. It can become reducible for specific values of $\theta_{\ell}$ in which case the order in which the tensor product \eqref{eq:H} is taken is important. Solving Wronskian Bethe equations produces $\multiplicity_{\lambda}$ for $\mathcal{H}$ as in \eqref{eq:H} regardless of whether it is reducible or not. So far it is an experimental fact which however was expected. Indeed, the isomorphism between Wronskian and Bethe algebras was proven in \cite{MTV} for $\gl_{\rank}$ case under the assumption of cyclicity of Yangian representation, cyclicity in the case of reducible representations was achieved by choosing an order in product \eqref{eq:H}. These observations explain the logic behind the conjecture below.

\begin{conjecture*}
Let $\mathcal{H}$ \eqref{eq:H} be a cyclic representation of $\Yangian(\so_{2\rank})$ with the highest-weight vector being a cyclic vector. Define the Bethe algebra as the commutative subalgebra of ${\rm End}(\mathcal{H})$ generated by the row-to-row transfer matrices (in all finite-dimensional representations in the auxiliary space) with periodic boundary conditions, and let $\mathfrak{B}$ be the Bethe algebra restricted to the highest-weight subspace (\wrt $\so_{2\rank}$ action) of $\multiplicity_\lambda \grep_\lambda\subset \mathcal{H}$ \eqref{eq:HilberSpaceDecomposition}. Then $\mathfrak{B}$ is a maximal commutative subalgebra of this subspace endomorphisms (completeness), and variables $c_{a,l}$, $a=0,\ldots,\rank$, that satisfy conditions \eqref{eq:FiniteDifferenceWronsk},\eqref{eq:W},\eqref{eq:SpinorFromVec} and  the gauge-fixing prescriptions for $\psi_a,\psi_{ab}$   form an algebra $\mathcal{W}$ (Wronskian Bethe algebra) isomorphic to $\mathfrak{B}$ (faithfulness). In particular, the algebraic number of solutions to \eqref{eq:FiniteDifferenceWronsk},\eqref{eq:W},\eqref{eq:SpinorFromVec} with the prescribed gauge fixing is $\multiplicity_{\lambda}$.
\end{conjecture*}
\noindent
Our general expectation is that the requirement for all $Q_{(a)}^{\mathfrak{i}}$ to satisfy the analytic Bethe anastz is the necessary and sufficient condition for faithfullness. 
This requirement is not part of the above conjecture because we expect it to be a consequence. As was discussed in Section~\ref{sec:VB},  it would be a consequence for generic values of parameters and it was a consequence for all attempted experiments as well. A weaker and safer version of the conjecture would be to include this requirement. From a practical perspective, checking it starting from the solution of Wronskian Bethe equations demands a small fraction of the computation time comparted to solving of Wronskian Bethe equations themselves.

\medskip
\noindent
In the case when $\mathcal{H}$ is reducible, and if one wants to select the solutions corresponding to the quotient-representation $\mathcal{H}/U$, where $U$ is an invariant subspace, one needs further constraints on $c_{a,l}$. Consider for instance a Kirillov-Reshetikhin module $\qrep_{1,m}(\theta)$ that is the {\it irrep} with $P_a=(u-\theta)^{\delta_{1,a}[m]_{\sh}}$, and build a spin chain by taking tensor products $\mathcal{H}'=\bigotimes_{\ell=1}^L \qrep_{1,m}(\theta_\ell)$. It is isomorphic to the quotient-representation inside spin chain \eqref{eq:H} of length $m\,L$ which is based on the same Drinfeld polynomials as $\mathcal{H}'$ but reducible. To select the Bethe algebra on the quotient one imposes the extra requirement: the polynomials $\frac{T_{1,s}}{\sigma_{1}^{[\rank-1+s]}\sigma_{1}^{[-\rank+1-s]}}$ are divisible by $ (p^{[\rank-2]}p^{[-\rank+2]})^{[m-s]_\sh}$, $s<m$, where  $p=\prod\limits_{\ell=1}^{L}(u-\theta_{\ell})$.

Kinematic constraints \eqref{eq:SpinorFromVec} serve a similar purpose. We can construct any spin chain as a quotient from a specially designed basic spin chain, and applying the kinematic constraints corresponds to performing this quotient. In particular, the $a$’th fundamental representation of Yangian with $P_a=(u-\theta)$, $1<a\leq \rank-2$, can be achieved as a quotient-representation in the vector spin chain of length $a$ described by $P_1=(u-\theta)^{[a]_\sh}$, this feature can be spotted in the structure of \rhs of \eqref{eq:W}. As for spinor representations, consider as an illustration the simplest case when $P_{\rank-1}$ and $P_{\rank}$ are not co-prime: $P_{\rank}=P_{\rank-1}=u-\theta$. This corresponds to an irreducible Yangian representation of dimension $2^{2r-2}$ which decomposes as an $\mathfrak{so}_{2\rank}$-module into the sum $\Lambda^{r-1}(\mathbb{C}^{2r})\oplus\Lambda^{r-3}(\mathbb{C}^{2r})\oplus\ldots$. This Yangian representation can be built as a quotient-representation in  the vector spin chain of length $r-1$ described by $P_1=(u-\theta)^{[r-1]_\sh}$.

\newpage
\section{Conclusions}
We put into practice the extended Q-system introduced in \cite{Ferrando:2020vzk} and \cite{Ekhammar:2020enr}. To this end, we compromise between maintaining a covariant description and keeping a manageable number of Q-functions: the full $\so_{2\rank}$ covariance is partially broken to $\gl_{\rank}$ subalgebra. Whereas this breakdown and an equivalent parameterisation is already  present in \cite{Ferrando:2020vzk}, we make additionally a geometric connection to the fused pure spinor point of view of \cite{Ekhammar:2020enr}. This allows us to single out a system of equations \eqref{eq:SpinorQSystem} that contains all the essential information and is useful for explicit computations. 

Across the paper we encountered three classes of functions: $V_{A},\Qq_{A},\Psi_{A}$ transforming covariantly under $\gl_\rank$ action. Their contra-variant counterparts are $V^{A},\Qq^{A},\Psi^{A}$. Let us summarise basic relationships between them. To focus on the conceptual features, we shall suppress shifts of the spectral parameter in $\mu_{AB}$ and set $\Phi=1$ (while the bulk of the article contains full expressions):
\be
\Qq^A=\epsilon^{AB}\Qq_B\,,\quad \Psi^A=\epsilon^{AB}\Psi_B\,,\quad V^A=\mu^{AB}V_{B}\,\quad
\Qq^a=V^a\,,\quad \Qq_a=\Psi_a.
\ee
While $\Qq^A=V^A$ in the suggested simplifications, we emphasise that $\bP_A\neq \Psi_A$ for $|A|>1$. These features reflect the following fact: $\bP^A,\bP_B$ form $\gl_{\rank}$ Q-system \eqref{eq:QQoriginal}, so do $V^A,V_B$ and in fact $V_{A}{}^B=g_{AI}V^{IB}$ is a piece of $\gl_{2\rank}$ Q-system, while, in contrast, $\Psi^A,\Psi_B$ form pure spinor Q-system \eqref{eq:PsiPsi}.

Equations \eqref{eq:SpinorQSystem} are formulated for the triple $\Psi_\es,\Psi_a,\Psi_{ab}$ which plays the fundamental role for this article. All other Q-functions can be computed from the triple through \eqref{eq:PureSpinorCondition} and \eqref{eq:VQPsi} and we proved that all the relations of the extended Q-system follow from \eqref{eq:SpinorQSystem}. This in turn implies that an extended Q-system can always be constructed from $Q_{(a)}^1$ and is unique up to the symmetry. The derivation of the extended Q-system was not easy to demonstrate previously: In \cite{Ferrando:2020vzk}, some of its relations were conjectured while derived for low-rank cases only; In \cite{Ekhammar:2020enr}, the relations were obtained using an analytic and non-combinatorial argument via ODE/IM. We now offer their  combinatorial derivation for an arbitrary rank. 

We demonstrated that, in general position, $\Psi_\es,\Psi_a,\Psi_{ab}$ and $Q_{(a)}^1$ can be derived from one another and, owing to the fact that zeros of $Q_{(a)}^1$ satisfy nested Bethe equations, we confirm that our approach is in principle consistent with these equations in general position. We however circumvented the passage to the nested Bethe equations: Section~\ref{sec:RationalSpinChains} demonstrates how to use directly Wronskian Bethe equations to find spectrum of rational spin chains in a variety of different representations. The proposed approach is more efficient, by a large margin, than solving nested Bethe equations, at least when the goal is to find all solutions instead of focusing on a single one. Furthermore, the offered equations define de-facto a coordinate ring of an algebraic variety, and we conjecture based on rich numerical evidence that this ring is isomorphic to the Bethe algebra which should be maximally commutative, notably the conjecture does not involve a general position assumption. Using nested Bethe equations towards the same type of statement is impractical if even possible. These observations are in full parallel with the ones for Wronskian Bethe equations of $\sl_{\rank}$ spin chains in \cite{MTV,Chernyak:2020lgw}.

Having additionally in mind potential applications to \eg TBA equations, we investigated in Section \ref{sec:Tfunc} how T-functions $T_{a,s}$ are decomposed into $\gl_{\rank}$-invariants. We got closed although rather bulky formulae for all symmetric powers of miniscule representations ($a=1,r-1,r$); many expressions were known previously, in particular in \cite{Ferrando:2020vzk}, but expressions involving fermionic representations ($a=r-1,r$) with $s>1$ are new. Arguments from representation theory, notably sign-free decompositions \eqref{eq:suggestive} and \eqref{eq:Tassp}, suggest that the obtained expressions are as compact as they can possibly be. In this respect, more elegant and explicitly $\so_{2\rank}$-invariant expressions \eqref{eq:Ts} seem to be more practical even though they use more Q-functions, especially because computing all the Q-functions is streamlined by \eqref{eq:PureSpinorCondition} and \eqref{eq:VQPsi}.

\begin{acknowledgments}
\vspace{-1em}
We would like to thank Luca~Cassia, Vladimir~Kazakov, Paul~Ryan, and Maor~Ben-Shahar for stimulating and interesting discussions.  This work was supported by the Knut and Alice Wallenberg Foundation under grant ``Exact Results in Gauge and String Theories''  Dnr KAW 2015.0083. 
\end{acknowledgments}
\appendix

\section{Implementing Weyl group action}
\label{sec:Weyl}
Let $U$ be an invertible linear transformation from the Clifford algebra acting on spinors and $\CO$ an invertible linear transformation acting on vectors. The transformation $\Psi \mapsto U\,\Psi\,,V \mapsto \CO\,V$ is a symmetry of the Q-system if it preserves inner products \eqref{eq:Quant} and if $U$ and $\CO$ are mutually compatible. To preserve the inner product means that $\CO\in \mathsf{O}(2r)$ while $U$ must satisfy $ U^TCU =  C$. Compatibility requires that $U^{-1}\Gamma^{i}U = \CO^{i}{}_j \Gamma^j$. The two constraints on $U$ imply that $U$ is an element of $\mathsf{Pin}(2\rank)$ and hence we can consider $\mathsf{Pin}(2\rank)$ as our global symmetry group.

Weyl group $\Weyl$ is the symmetry group acting on the space of weights, it preserves the set of weights for any finite-dimensional representation of $\so_{2\rank}$. If $\weyl\in\Weyl$, we call $\hweyl\in\mathsf{Pin}(2\rank)$ a representative of $\weyl$ if, in any given representation $L(\omega)$ the following holds:
\be
\hweyl\cdot \basvec_{\gamma}=\#\,\basvec_{\weyl\cdot\gamma}\,,
\ee
where $\basvec_{\gamma}\in L(\omega)$ is a basis vector of weight $\gamma$ and $\#$ is a numerical constant that can depend on $\weyl$ and $\gamma$. Except in special cases, it is impossible to set all such constants to identity or to make $\hweyl$ a homomorphism from $\Weyl$ to $\mathsf{Pin}(2\rank)$. Nevertheless, since $\hweyl\in\mathsf{Pin}(2\rank)$, $\hweyl$ is obviously a symmetry transformation of the extended Q-system and therefore, allowing some freedom of speech, we say the Weyl group is the symmetry of the extended Q-system as well.  This Weyl symmetry is another important organisational principle of the extended Q-system, note for instance its role in the original QQ-relations \eqref{eq:DynkinDiagramQQ}. To work with it, it is pertinent to fix $\#$. There are prescriptions with simple choices of  $\#$, and the goal of this appendix is to explicitly write them down. It is enough to do it for minuscule representations, the rest follows easily.

\medskip
\noindent
One choice for fixing $\#$ was proposed in \cite{Ekhammar:2020enr} using general Lie algebra formalism: For $\elweyl_a$ being the elementary Weyl reflection, its representative is picked as $\helweyl_{a} = e^{e_{a}}e^{-f_{a}}e^{e_{a}}$ where $e_{a}$ and $f_a$ are Chevalley generators. For the spinor representation we take, cf. Section~\ref{sec:Cartandecomp}, $e_{a} = \theta^{a+1}\partial_{a}\,, f_{a}=\theta^{a}\partial_{a+1}$ for $a<r$ and $e_{r} = \partial_{r}\partial_{r-1}\,,  f_{r}=\theta^{r-1}\theta^{r}$. Then, when acting on basis vectors of fermionic representations, the explicit action computes to
\begin{subequations}
\label{prescription0}
\begin{align}
    \helweyl_a\cdot\theta^{A,a} &= \theta^{A,a+1} \,,
    &
    \helweyl_{a}\cdot \theta^{A,a+1} &= -\theta^{A,a}\,,
     &
    \helweyl_r\cdot \theta^B&=\theta^{B,r-1,r}\,, 
    & 
    \helweyl_r\cdot\theta^{B,r-1,r} &= -\theta^{B}\,,
\end{align}
and the consistent transformation for basis vectors for vector representation is
\begin{align}
\helweyl_a\cdot \mathbf{e}_{\pm a} &=  \mathbf{e}_{\pm (a+1)} \,,
    &
\helweyl_a \cdot \mathbf{e}_{\pm (a+1)}  &= - \mathbf{e}_{\pm a} \,,
    &
\helweyl_r \cdot \mathbf{e}_{\pm(r-1)} &= \mathbf{e}_{\mp r}\,, 
 &
\helweyl_r \cdot \mathbf{e}_{\pm r} &= -\mathbf{e}_{\mp (r-1)}\,,
\end{align}
\end{subequations}
In these formulae, $a=1,2,\ldots,\rank-1$, $A$ does not contain $a$ and $a+1$ and $B$ does not contain $r$ and $r-1$; only nontrivial transformations were written down. 

It is useful to tabulate more transformations than just the elementary Weyl reflections. We will here describe explicitly permutations of any two indices dubbed $\sigweyl_{ab}$ and sign swaps dubbed $\sigweyl_a$. Their action on the orthonormal basis of the weight space is
\begin{eqnarray}
\sigweyl_{ab}\cdot \varepsilon_a =\varepsilon_b\,,\quad
\sigweyl_{ab}\cdot \varepsilon_b =\varepsilon_a\,,\quad
\sigweyl_{ab}\cdot \varepsilon_c &=&\varepsilon_c\,,\ 
c\neq a,b\,,
\\
\sigweyl_a\cdot \varepsilon_a =-\varepsilon_a\,,\quad\quad\quad\quad
\sigweyl_a\cdot \varepsilon_b &=&\varepsilon_b\,,\ b\neq a\,.
\end{eqnarray}
Reflections $\sigweyl_a$ are not elements of Weyl group, their products $\sigweyl_a\sigweyl_b$ are. Nevertheless,  $\sigweyl_a$ are symmetries for weights in all vector representations and for Dirac spinor, so considering them is still useful.

Note that $\elweyl_a=\sigweyl_{a,a+1}$ for $a<r$ and $\elweyl_r=\sigweyl_{r-1,r}\sigweyl_{r-1}\sigweyl_{r}$. In \cite{Ekhammar:2020enr}, defining $\hweyl$ for $\weyl$ other than  $\elweyl_a$ was based on representing $\weyl$ as a minimal length product of $\elweyl_a$. Given specifics of $\mathfrak{so}_{2\rank}$, it is more suitable to use a different convention for $\hat\sigweyl_{ab}$ and $\hat\sigweyl_a$ which symmetrically treats indices, this means that sign prescriptions in $\#$ shall be not always the same as implied by the recipe of \cite{Ekhammar:2020enr}.

For permutations, one can pick $\hat \sigweyl_{ab} = \Gamma^b \Gamma_a -\Gamma^a\Gamma_{b}+ \Gamma^b \Gamma^a\Gamma_a \Gamma_b +\Gamma_{a}\Gamma_{b}\Gamma^{b}\Gamma^{a}$. On the level of components this translates into
\begin{subequations}
\begin{align}
\label{prescriptionA}
    &\hat \sigweyl_{ab}: \Psi_{Aab} \mapsto \Psi_{Aab}\,,
    &
    &\hat \sigweyl_{ab}: \Psi_{Aa} \mapsto -\Psi_{Ab}\,,
    &
    &\hat \sigweyl_{ab}: \Psi_{Ab} \mapsto \Psi_{Aa}\,,
    &
    &\hat \sigweyl_{ab}: \Psi_{A} \mapsto \Psi_A\,, \\
    &\hat \sigweyl_{ab}: V^{\pm a} \mapsto -V^{\pm b}\,,
    &
    &\hat \sigweyl_{ab}: V^{\pm b} \mapsto V^{\pm a}\,.
    &
    &\hat \sigweyl_{ab}: V^{i} \rightarrow V^{i}\,,
\end{align}
\end{subequations}
where $A$ does not contain $a,-a$ and $i\neq \pm a,b$.

For sign flips, a good choice is $\hat \sigweyl_a = (-1)^N(\Gamma_a-\Gamma^{a})$ where $(-1)^{N}\theta^{A} = (-1)^{|A|}\theta^{A}$ and explicitly in Gamma-matrices $(-1)^{N}=2^{\rank}\, \Gamma_{1}{}^1\dots \Gamma_{\rank}{}^{\rank}$.  The action on components is
\begin{align}
\label{signflips}
    &\hat\sigweyl_a: \Psi_{A} \mapsto \Psi_{Aa}\,,
    &
    &\hat\sigweyl_a: \Psi_{Aa} \mapsto \Psi_{A}\,.
    &\hat\sigweyl_a: V^a \mapsto V_a\,,
    &
    &\hat\sigweyl_a: V_a \mapsto V^a\,,
    &
    &\hat\sigweyl_a: V^i \mapsto V^i\,,
\end{align}
where $A$ does not contain $a$ and  $i\neq \pm a$.

We note that clearly the action on $V^i$ is a reflection which is in the disconnected component from identity of $\mathsf{O}(2\rank)$. Another remark is that $(-1)^{N}=2^{\rank}\, \Gamma_{1}{}^1\dots \Gamma_{\rank}{}^{\rank}$ is recognisable as the chiral gamma-matrix $\Gamma^5$, however in the chosen normalisation it is not an element of $\mathsf{Pin}(2r)$, the correctly normalised option is $\Gamma^5=e^{-\frac{i\pi}{2}\rank}(-1)^N$; under the map $\mathsf{Pin}(2r)\to \mathsf{O}(2r)$, $\Gamma^5$ and $-\Gamma^5$ have their image in the connected component $\mathsf{SO}(2r)$ as $-{\rm Id}$ which is nothing but the space inversion $V^i\mapsto -V^i$.

\medskip
\noindent
Prescription \eqref{prescriptionA} is compatible with \eqref{prescription0} for what concerns the fundamental reflections $\elweyl_a$, however this comes at price that $\hat\sigweyl_{ba}=\hat\sigweyl_{ab}^{-1}\neq\hat\sigweyl_{ab}$. If not to insist on \eqref{prescription0}, a more symmetric prescription is possible with $\hat\sigweyl_{ab}=\hat\sigweyl_{ba}=\ii(\Gamma^a\Gamma_{b}+\Gamma^b \Gamma_a + \Gamma^a \Gamma^b\Gamma_a \Gamma_b +\Gamma_{a}\Gamma_{b}\Gamma^{b}\Gamma^{a})$ which explicitly translates in the action on components as 
\begin{subequations}
\label{prescriptionB}
\begin{align}
    &\hat \sigweyl_{ab}: \Psi_{Aab} \mapsto -\ii\,\Psi_{Aab}\,,
    &
    &\hat \sigweyl_{ab}: \Psi_{Aa} \mapsto \ii\,\Psi_{Ab}\,,
    &
    &\hat \sigweyl_{ab}: \Psi_{Ab} \mapsto \ii\,\Psi_{Aa}\,,
    &
    &\hat \sigweyl_{ab}: \Psi_{A} \mapsto \ii\,\Psi_A\,, \\
    &\hat \sigweyl_{ab}: V^{\pm a} \mapsto V^{\pm b}\,,
    &
    &\hat \sigweyl_{ab}: V^{\pm b} \mapsto V^{\pm a}\,.
    &
    &\hat \sigweyl_{ab}: V^{i} \rightarrow V^{i}\,.
\end{align}
\end{subequations}
For what concerns vectors, all $\#$ factors in this prescription are set to $1$, so \eqref{prescriptionB} and \eqref{signflips} realise an embedding of $\Weyl$ as subgroup of $\mathsf{SO}(2\rank)$, and $\Weyl$ adjoined with $\sigweyl_a$ is embedded in $\mathsf{O}(2\rank)$.

\medskip
\noindent
 The discussed Weyl symmetry properties allow for further comparison with results of \cite{Ferrando:2020vzk}. Most of the relations there  enjoy both $\gl_{\rank}$ covariance and covariance with respect to Weyl group action. This is expected as \cite{Ferrando:2020vzk} works with a Q-system on the Weyl orbit. In hindsight, as we have typically the same covariance structures, many of our expressions have a counterpart in \cite{Ferrando:2020vzk}, note however that some of the relations in \cite{Ferrando:2020vzk} were not derived but conjectured for arbitrary rank, for instance an equivalent of fused orthogonality  \eqref{eq:fusedortho}. To relate formulae between the two papers, we note that spinor Q-functions are denoted in \cite{Ferrando:2020vzk} as $S_{\{i_1,\ldots,i_{\rank}\}}$, and the relation to $\Psi_{A}$ is the following: $\Psi_{A}= S_{\{A'\bar{A}\}}$, where $A'$ means a multi-index with all entries having swapped signs, and ${\bar A}$ means the multi-index with entries from the complementary set. Q-functions for the tensor nodes of Dynkin diagram are identified as $V^{I}\ ({\rm here})=Q_{I}\ ({\rm in}\ \text{\cite{Ferrando:2020vzk}})$. The offered identifications might need an adjustment in normalisations: \cite{Ferrando:2020vzk} tends to take twist variables ($\tau_{a,{\rm in\ [8]}}=x_{a,{\rm here}}$) and source terms of Bethe equations outside of Q-functions whereas we pack them all inside, in particular through the dressing factors, see Section~4.3 of \cite{Ekhammar:2020enr}.

Among the most notable identifications with \cite{Ferrando:2020vzk}, we mention that (5.10) there is the same as our \eqref{eq:VectorUpperFromSpinor} when $I$ in $Q_I$ has only positive entries. A particular case $I=A=\emptyset$ is the Wronskian condition \eqref{eq:DeterminantCondition} present in \cite{Ferrando:2020vzk} as (7.2). There are also many equivalent statements when it comes to the expressions for T-functions,  we comment on that in Section~\ref{sec:Tfunc}.

The key difference with \cite{Ferrando:2020vzk} comes in the recognition that $\Psi_{A}$ are coordinates of pure spinors. This observation, originally from \cite{Ekhammar:2020enr}, allows us to use Cartan formulae \eqref{eq:PureSpinorCondition} and reduce all the analysis to $\Psi_{\emptyset},\Psi_a,\Psi_{ab}$ which is the main working tool for achieving novel results in our work (they are summarised in Conclusions). To our understanding, an analogous observation is missing from \cite{Ferrando:2020vzk}. For instance, if $I$ has negative entries, (5.10) shall involve some $\Psi_{B}$ with $|B|\geq 2$. It is unclear how to practically use (5.10) then, whereas we derive instead equation \eqref{eq:VQPsi} which offers even a way to compute all and not only those that are on the Weyl orbit. So, even though somewhat in disguise, \eqref{eq:VQPsi} restores the full $\so_{2\rank}$ symmetry of the formalism.
\section{Technical details}
\label{app:techproofs}
\begin{proof}[{\bf Proof of Lemma~\ref{leamma:purespinor}}] Due to rescaling invariance of equations and condition $\Psi_{\emptyset}\neq 0$, we shall assume $\Psi_{\emptyset}=1$ without loss of generality.

We need to prove that $\Psi_{(k)}$ computed by Cartan formulae \eqref{eq:PureSpinorCondition} solve equations of pure spinor Q-system \eqref{eq:PsiPsi}. To demonstrate the statement for $\Psi_{(k)}$, we shall recursively assume that it was already proven for all $\Psi_{(k')}$ with $k'<k$. The statement is obvious for $k=1,2$ which is the base for the induction. 

By the nature of equation  \eqref{eq:PsiPsi} which involves only a subset of indices $Aab$, we can assume that $|Aab|=k=r$, so $\Psi_{(k)}$ is the top form and $\Psi_{Aab}=*\Psi_{(k)}\,\epsilon_{Aab}$, where $*$ is Hodge operation. We can rewrite then \eqref{eq:PsiPsi} as
\be\label{eq:A1}
W(*\Psi_{(k)},*(\ba\wedge\bb\wedge\Psi_{(k-2)}))=W(*(\ba\wedge\Psi_{(k-1)}),*(\bb\wedge\Psi_{(k-1)}))
\ee
that should hold for arbitrary vectors $\ba,\bb$. Furthermore, working with top forms enables us to use Plücker identity
\be
*({\bf v}_1\wedge{\bf v}_2\ldots\wedge{\bf v}_k)\,{\bf a}=\sum_{i=1}^k *({\bf v}_1\wedge{\bf v}_2\ldots\wedge\underset{i{\rm-th\ position}}{{\bf a}}\wedge\ldots\wedge{\bf v}_k)\,{\bf v}_i\,.
\ee

Details of the computation depends on the parity of $k$, we shall consider the two cases separately.

\medskip
\noindent{\it Even case, $k=2n+2$.} First, derive the important technical result:
\be
\label{eq:diffrelation}
\Psi_{(2n+2)}^+-\Psi_{(2n+2)}^-=(\Psi_{(2)}^+-\Psi_{(2)}^-)\wedge\Psi_{(2n)}^{\pm}=\Psi_{(1)}^+\wedge\Psi_{(1)}^-\wedge\Psi_{(2n)}^{\pm}.
\ee
Derivation is done by induction in $n$. For $n=0$, the first equality is tautological and the second is $$\Psi_{(2)}^+-\Psi_{(2)}^-=\Psi_{(1)}^+\wedge\Psi_{(1)}^-$$ which is nothing but \eqref{eq:U1Conditions}. We also note that, obviously, $\Psi_{(1)}^{\pm}\wedge(\Psi_{(1)}^+\wedge\Psi_{(1)}^-)=0$ and then $\Psi_{(1)}^{\pm}\wedge\Psi_{(2n)}^+=\Psi_{(1)}^{\pm}\wedge\Psi_{(2n)}^-$ because $\Psi_{(2n)}=\frac 1{n!}\Psi_{(2)}^n$, this explains the possibility to choose an arbitrary shift sign in \eqref{eq:diffrelation}.

For $n>0$, the derivation of \eqref{eq:diffrelation} proceeds as follows 
\begin{align}
\nonumber
&\Psi_{(2n+2)}^+ =\frac{1}{n+1}\Psi_{(2)}^+\wedge\Psi_{(2n)}^+=\frac{1}{n+1}\Psi_{(2)}^+\wedge\Psi_{(2n)}^-+\frac{1}{n+1}(\Psi_{(2)}^+-\Psi_{(2)}^-)\wedge\Psi_{(2n-2)}^{\pm}\wedge{\Psi_{(2)}^{\pm}}
\\
&=\Psi_{(2n+2)}^-+\frac{1}{n+1}(\Psi_{(2)}^+-\Psi_{(2)}^-)\wedge\Psi_{(2n)}^{-}+\frac{n}{n+1}(\Psi_{(2)}^+-\Psi_{(2)}^-)\wedge\Psi_{(2n)}^{\pm}=\Psi_{(2n+2)}^-+(\Psi_{(2)}^+-\Psi_{(2)}^-)\wedge\Psi_{(2n)}^{\pm}\,.
\nonumber
\end{align}

We use \eqref{eq:diffrelation} to rewrite the \lhs of \eqref{eq:A1} 
\begin{align}
\label{eq:A4}
W(*\Psi_{(2n+2)},*(\ba\wedge\bb\wedge\Psi_{(2n)}))=*(\Psi_{(1)}^+\wedge\Psi_{(1)}^-\wedge\Psi_{(2n)}^{\pm})\,*\!(\ba\wedge\bb\wedge\Psi_{(2n)}^-)-*\Psi_{(2n+2)}^-\,*\!(\ba\wedge\bb\wedge\Psi_{(1)}^+\wedge\Psi_{(1)}^-\wedge\Psi_{(2n-2)}^{\pm})\,,
\end{align}
and in what follows we pick the negative shift $\Psi^-$ in both places where $\Psi^{\pm}$ appears. 

Consider now $\Psi_{ab}^-$ as components of an anti-symmetric matrix. Anti-symmetric matrices of even dimension have eigenvalues coming in pairs, where in each pair the two eigenvalues sum up to zero. Denote hence the eigenvalues by $\lambda_{\pm\alpha}\equiv\pm\lambda_{\alpha}$ and the corresponding eigenvectors as $\Upsilon_{\pm\alpha}$. Note that two eigenvectors $\Upsilon_{\alpha}$ and $\Upsilon_{\beta}$ are orthogonal (with naive metric) if $\alpha+\beta\neq 0$. With these remarks, it is easy to conclude that, with an appropriate normalisation of the eigenvectors, one can represent $\Psi_{(2)}^-=\sum\limits_{\alpha=1}^{n+1}\Upsilon_{\alpha}\wedge \Upsilon_{-\alpha}$. From here
\begin{align*}
    \Psi_{(2n+2)} &=\prod_{\gamma}\Upsilon_{\gamma}\wedge \Upsilon_{-\gamma}\,,&
    \Psi_{(2n)} &=\sum_{\alpha}\prod_{\gamma\neq\alpha}\Upsilon_{\gamma}\wedge \Upsilon_{-\gamma}\,,&
    \Psi_{(2n-2)}&=\sum_{\alpha<\beta}\prod_{\gamma\neq\alpha,\beta}\Upsilon_{\gamma}\wedge \Upsilon_{-\gamma}\,.
\end{align*}
We substitute these expressions into \eqref{eq:A4} and employ \Plucker relations: for vector $\ba$ in the first term and for vector $\bb$ in the second term. All the unwanted terms (those coming from the exchange of $\ba$ or $\bb$ with $\Upsilon_{\pm\gamma}$) shall cancel and what remains is precisely the \rhs of \eqref{eq:A1}, if we use $\Psi_{(2n+1)}^+=\Psi_{(1)}^{+}\wedge\Psi_{(2n)}^{\pm}$ and $\Psi_{(2n+1)}^-=\Psi_{(1)}^{-}\wedge\Psi_{(2n)}^{\pm}$.

\medskip
\noindent{\it Odd case, $k=2n+1$.} Using the above-established properties, we rewrite both sides of \eqref{eq:A1} as follows
\begin{subequations}
\label{eq:A5}
\begin{align}
W(*\Psi_{(2n+1)},*(\ba\wedge\bb\wedge\Psi_{(2n-1)})) &= *(\Psi_{(1)}^+\wedge\Psi_{(2n)}^-)\,*\!(\ba\wedge\bb\wedge\Psi_{(1)}^-\wedge\Psi_{(2n-2)}^-)-*(\Psi_{(1)}^-\wedge\Psi_{(2n)}^-)\,*\!(\ba\wedge\bb\wedge\Psi_{(1)}^+\wedge\Psi_{(2n-2)}^-)\,,
\\
W(*(\ba\wedge\Psi_{(2n)}),*(\bb\wedge\Psi_{(2n)})) &=*(\ba\wedge\Psi_{(1)}^+\wedge\Psi_{(1)}^-\wedge\Psi_{(2n-2)}^-)\,*\!(\bb\wedge\Psi_{(2n)}^-)-*(\ba\wedge\Psi_{(2n)}^-)\,*\!(\bb\wedge\Psi_{(1)}^+\wedge\Psi_{(1)}^-\wedge\Psi_{(2n-2)}^-)\,.
\end{align}
\end{subequations}
In odd $(2n+1)$-dimensions, a two-form can be represented as $\Psi_{(2)}^-=\sum\limits_{\alpha=1}^{n}\Upsilon_{\alpha}\wedge \Upsilon_{-\alpha}$, here one eigenvector of $\Psi_{ab}^-$ does not participate in the sum, the one with zero eigenvalue which is orthogonal to all other eigenvectors. Consequently, we have
\begin{align*}
&\hphantom{eeeee}&
    \Psi_{(2n)} &=\prod_{\gamma}\Upsilon_{\gamma}\wedge \Upsilon_{-\gamma}\,,&
    \Psi_{(2n-2)}&=\sum_{\alpha}\prod_{\gamma\neq\alpha}\Upsilon_{\gamma}\wedge \Upsilon_{-\gamma}\,.
&\hphantom{eeeee}&
\end{align*}
We substitute these expressions to both lines of \eqref{eq:A5} and then demonstrate that these lines are equal by employing \Plucker identity in the first line, with vector $\ba$ for the first summand and vector $\bb$ for the second. This accomplishes the proof that Ansatz \eqref{eq:PureSpinorCondition} solves \eqref{eq:PsiPsi}.

 To demonstrate that \eqref{eq:PureSpinorCondition} is the unique solution, assume the opposite: let there exists $\Psi_{(k)}'$ that is different from $\Psi_{(k)}$ and that also solves \eqref{eq:PsiPsi}. Consider the smallest $k$ for which such feature holds. Construct $\delta\Psi_{(k)}=\Psi_{(k)}'-\Psi_{(k)}$. Since $k\geq 3$ we can parameterise the components of $\delta\Psi_{(k)}$ as $\delta\Psi_{Aabc}$. As both $\Psi_{(k)}'$ and $\Psi_{(k)}$ solve \eqref{eq:PsiPsi}, it must be that $W(\delta\Psi_{Aabc},\Psi_{Aa})=0$ and so $\delta\Psi_{Aabc}=h_{bc}\Psi_{Aa}$, where $h_{bc}^+=h_{bc}^-$. On the other hand, $\delta\Psi_{Aabc}$ is anti-symmetric in its indices and so it must hold $\delta\Psi_{Aabc}=h_{ca}\Psi_{Ab}$. Then $0=W(\delta\Psi_{Aabc},\Psi_{Aa})=h_{ca}^\pm\, W(\Psi_{Ab},\Psi_{Aa})$ and we arrive at contradiction with the assumption of the lemma.
\end{proof}

\begin{proof}[{\bf Derivation of \eqref{eq:VectorSpinorRelationsb} for the proof of Theorem~\ref{thm:derivingextendedQsystem}}]
Recall that by definition $V^{i_1\dots i_\rank} = W(V^{i_1},\dots,V^{i_\rank})$, our aim is to establish the identification $V^{i_1\dots i_\rank}=\Psi^{-}C\Gamma^{i_1,\dots,i_\rank}\Psi^+$. The calculation will be divided into 3 cases: only positive indices, one negative index, and two or more negative~indices.

We start with the case when all indices are positive, then, using the definition of $\Gamma$ and $C$, we find $\Psi^{-}C\Gamma^{1\dots r}\Psi^{+} = \Psi_{\es}^{[2]_D}$. It is also true that $V^{1\dots r} = \Psi_{\es}^{[2]_D}$ as follows from $V^{1\dots r} V^{1\dots r-2}= W(V^{12\dots r-1},V^{1\dots,r-2,r})$ and $V^{a_1\dots a_{\rank-1}} = \epsilon^{a_1\dots a_{\rank-1}a_{\rank}}\frac{\bP_{a_{\rank}}}{\Phi^2}$. We conclude that the highest-weight components agree.

Next, let one superscript in $V^{i_1,\dots,i_r}$ be negative, for convenience we will write $V^{Aa}{}_{b}$ with $|A|=\rank-2$. The reason to introduce the extra index $a$ is the following $\Gamma$-matrix identity: $\Gamma^{Aa}{}_{b} = \Gamma^{Aa}\Gamma_{b} - \frac{1}{2}\Gamma^{[A}\delta^{a]}_{b}$ where $[\,\,]$ denotes anti-symmetrisation. Using this we find
\begin{equation}
\begin{split}
    \Psi^{-}C\Gamma^{Aa}{}_{b}\Psi^{+} &= \mu^+_{bc}V^{Aac} + \bP^+_{b}\Psi^{-}C\Gamma^{Aa}\Psi^{+} -V^{[A}\delta^{a]}_{b}
    =\mu_{bc}^+ V^{Aac}+\Psi^{[2]_D}_{\es}\epsilon^{Aac}\bP^+_{b}\bP^-_{c}-V^{[A}\delta^{a]}_{b}\,.
\end{split}
\end{equation}
Notice that $\Psi_{\es}^{[2]_D}\epsilon^{Aac}\bP_{bc} = V^{B}\epsilon^{Aac}\epsilon_{Bbc} = V^{[A}\delta^{a]}_{b}$ so that we have
\begin{equation}\label{eq:FinalOneIndexV}
    \Psi^{-}C\Gamma^{Aa}{}_{b}\Psi^{+} = \mu_{bc}^+V^{Aac}+\Psi_{\es}^{[2]_D} \epsilon^{Aac}\bP_{c}^+\bP_{b}^-\,.
\end{equation}
From the definition of $V^{Aa}{}_{b}$ as a Wronskian of vectors $V^i$ it follows that
\begin{equation}
    V^{Aa}{}_{b} = \mu^+_{bc}V^{Aac}-(V^{Aa})^+\bP_{bc}(V^{c})^{[-r+1]} = \mu^+_{bc}V^{Aac}+\Psi^{[2]_D}_{\es}\epsilon^{Aac}\bP_{c}^+\bP^-_{b} = \Psi^{-}C\Gamma^{Aa}{}_{b}\Psi^{+}\,.
\end{equation}
This completes the case $V^{Aa}{}_{b}$.

Finally we consider the case when two or more indices are lowered, the notation $V^{Aa}{}_{Bbc}$ with $|A|+|B|=\rank-3$ will be used. We proceed just as before and start by finding an identity for $\Gamma$-matrices that allows us to raise indices. Multiplying \eqref{eq:PSDifferential} from the left with additional differential operators and reordering leads to
\begin{align}\label{eq:DifferentialGeneralized}
    &\Gamma_{Bbc}\Psi = (\mu_{Bbc,C}\Gamma^{C}+\mu_{[Bb,|D|}\bP_{c]}\Gamma^{D}\Gamma^+-\mu_{[B,|E|}\mu_{bc]}\Gamma^{E})\Psi+\dots\,,
\end{align}
where all suppressed terms contain $\Gamma$-matrices with less than $|B|$ indices. In the above equation, indices surrounded by $|\,\,|$ are not to be anti-symmetrized with the others enclosed by $[\,\,]$. Using this as well as  $\Gamma^{Aa}{}_{Bbc} = \Gamma^{Aa}\Gamma_{Bbc}-\frac{1}{2}\Gamma^{[A}{}_{[cB}\delta^{a]}_{b]}+\dots$ leads to
\begin{equation}\label{eq:VrankPsiExpression}
    \Psi^{-}C\Gamma^{Aa}{}_{Bbc}\Psi^{+} = \mu^+_{Bbc,C}V^{AaC}+\mu^+_{[Bb,|D|}\bP^+_{c]}\Psi^-C\Gamma^{AaD}\Psi^{+}-V^{Aa}{}_{[B}\mu^+_{bc]} - V^{[A}{}_{[cB}\delta^{a]}_{b]}\,.
\end{equation}
All terms previously omitted in \eqref{eq:DifferentialGeneralized} have now dropped out due to the projection relations. The second term can be simplified according to $\mu^{+}_{[Bb,|D|}\bP^+_{c]}\Psi^{-}C\Gamma^{AaD}\Psi^{+} = \frac{\Psi^{+}_{\es}}{\Psi^{-}_{\es}}\mu^+_{[Bb,|D|}\bP_{c]}V^{AaD}_{-} = \Psi_{\es}^{[2]_D}\epsilon^{AaDd}\mu^+_{[Bb,|D|}\bP^+_{c]}\bP^-_{d}$. Notice that
\begin{equation}
    \Psi_{\es}^{[2]_D} \epsilon^{AaDd}\mu^+_{[Bb,|D|}\bP_{c]d} = \epsilon^{AaDd}\mu^+_{[Bb,|D}\epsilon_{F|c]d}V^{F} = V^{Aa}{}_{[B}\mu^{+}_{bc]} + V^{[A}{}_{[cB}\delta^{a]}_{b]}\,.
\end{equation}
So that in the end one can write
\begin{equation}
    \Psi^{-}C\Gamma^{Aa}{}_{Bbc}\Psi^{+} = \mu^+_{Bbc,C}V^{AaC}+\Psi_{\es}^{[2]_D}\epsilon^{AaDd}\mu^+_{[Bb,|D|}\bP^+_{c]}\bP^-_{d}\,.
\end{equation}
From the definition of $V^{Aa}{}_{Bbc}$ as Wronskian determinant of vectors, one can calculate
\begin{equation}
\begin{split}\label{eq:VrankUpperLower}
    V^{Aa}{}_{Bbc} &=  
    \mu^{+}_{Bbc,C}V^{AaC} -(V^{Aa}{}_{[Bb})^+ \bP_{c]d}(V^{d})^{[-r+1]} 
    =\mu^+_{Bbc,C}V^{AaC} + \epsilon^{AaDd}\Psi_{\es}^{[2]_D}\mu^+_{[Bb,|D|}\bP^-_{c]}\bP_{d}^+
\end{split}
\end{equation}
showing that indeed $\Psi^-C\Gamma^{Aa}{}_{Bbc}\Psi^+ = V^{Aa}{}_{Bbc}$.

\end{proof}
\begin{proof}[{\bf Proof of Lemma~\ref{lemma:sufficient}}]\ 
By rescaling invariance, we can replace $\Psi_{(k)}$ with $\psi_{(k)}$ in \eqref{eq:PureSpinorCondition} and so all $\psi_{A}$ with $|A|>2$ are rational combinatinos of polynomials $\psi_{0},\psi_{a},\psi_{ab}$, with only $\psi_0$ in the denominators. Hence the only potential way to violate analytic Bethe Ansatz for fermionic Q-functions is to get a pole in $\psi_{(k)}$ located at a zero of $\psi_{0}$. 

For forms of odd rank, one has $\psi_{(2n+1)}^+\propto \psi_{(1)}^+\wedge\left(\left(\frac{\psi_{(2)}}{\psi_0}\right)^n\right)^\pm$ and therefore if $\psi_{(2n+1)}^+$ has a pole at a point $u^*$, it must be that both $\psi_0^+$ and $\psi_0^-$ vanish at this point or in other words $\psi_0$ has zeros at $u^*+\ii/2$ and $u^*-\ii/2$. 

For forms of even rank, consider \eqref{eq:diffrelation} and restore $\Psi_\es$ in it. Rewriting it for polynomial functions, one gets
\be
\psi_{(2n+2)}^+\psi_0^--\psi_{(2n+2)}^-\psi_0^+\propto \psi_{(1)}^+\wedge\psi_{(1)}^-\wedge \left(\left(\frac{\psi_{(2)}}{\psi_0}\right)^n\right)^\pm\,.
\ee
If $\psi_{(2n+2)}^+$ has a pole but $\psi_{(2n+2)}^+\psi_0^-$ does not this means that $\psi_0^+$ and $\psi_0^-$ have zeros at the same point. The same applies to the second term on the \lhs of the above relation. If both terms on the \lhs have a pole at the same point, we again come to the same conclusion about zeros of $\psi_0^+$ and $\psi_0^-$. Finally, if at a given point only one term on the \lhs has a pole then the \rhs has a pole as well and we come back to the argument used for $\psi_{(2n+1)}^+$.

For vector Q-functions, the analytic Bethe Ansatz reads $V^I\propto \sigma_{|I|}v^{I}$, where $v^{I}$ are polynomials. For the case of positive indices only, the polynomiality of $v^A$ is ensured by \eqref{eq:SpinorFromVec}.  For arbitrary $I$, we use $v_{A}{}^B=\mu_{AA'}^{[m]}v^{A'B}$. Notice now that $\mu_{ab}=\frac{\Psi_{ab}}{\Psi_\es}=\frac{\psi_{ab}}{\psi_0}$ is actually a rational function of the spectral parameter with only possible poles at zeros of $\psi_0$, and notice that there are at least two different admissible values of $m$ separated by $1$ for $|A|+B|\leq r-2$ which covers all the cases of tensor $Q_{(a)}^{\mathfrak{i}}$ on Dynkin diagram. By varying $m$, we get again to the conclusion that poles in $v^{I}$---the only way to violate analytic Bethe Ansatz in the setting of the lemma---would imply existence of zeros in $\psi_0$ separated by $\ii$.
\end{proof}

\section*{Conflict of interest statement}
\vspace{-1em}
On behalf of all authors, the corresponding author states that there is no conflict of interest.
\bigskip
\bibliography{apssamp}

\end{document}